\documentclass[
    aps,
    pra,
    longbibliography,
    superscriptaddress,
    amsmath,
    amssymb,
    amsfonts, 
    tightenlines,
    twocolumn,  
    ]{revtex4-2} 

\usepackage{graphicx} 
\usepackage[utf8]{inputenc}
\usepackage{amssymb}
\usepackage{float}
\usepackage{amsmath}
\usepackage{amsthm}
\usepackage{amsfonts}
\usepackage{cancel}
\usepackage{physics}
\usepackage{bbold}
\usepackage{bm}
\usepackage{braket}
\usepackage{graphicx}
\usepackage[title]{appendix}%
\usepackage{xcolor}
\usepackage[colorlinks=true,citecolor=teal,linkcolor=teal,urlcolor=teal]{hyperref}

\usepackage{charter}
\usepackage{times}

\newcommand{\nn}{\nonumber}

\newcommand{\EV}[1]{\left < #1 \right >}
\newcommand{\bk}[1]{\left ( #1\right )}

\newcommand{\eqn}[1]{\begin{eqnarray} \newline #1 \end{eqnarray}}
\newcommand{\ee}{&=&}
\newcommand{\hs}{\hspace{0.2cm}}
\newcommand{\adag}{\hat{a}^{\dag}}

\definecolor{darkgreen}{rgb}{0.0, 0.42, 0.24}

\definecolor{rsblue}{rgb}{0.0, 0.5, 0.69}

\newtheorem{definition}{Definition}
\newtheorem{theorem}{Theorem}
\newtheorem{lemma}{Lemma}
\newtheorem{corollary}{Corollary}
\begin{document}


\title{Certification of linear optical quantum state preparation}

\author{Riko Schadow}
\affiliation{Dahlem Center for Complex Quantum Systems, Freie Universit{\"a}t Berlin, 14195 Berlin, Germany}

\author{Naomi Spier}
\affiliation{MESA+ Institute for Nanotechnology, University of Twente, 7500 AE Enschede, The Netherlands}

\affiliation{Department of Applied Physics and
Science Education, Eindhoven
University of Technology, P. O. Box 513, 5600 MB Eindhoven, The
Netherlands}

\author{Stefan N. van den Hoven}
\affiliation{MESA+ Institute for Nanotechnology, University of Twente, 7500 AE Enschede, The Netherlands}

\author{Malaquias Correa Anguita}
\affiliation{MESA+ Institute for Nanotechnology, University of Twente, 7500 AE Enschede, The Netherlands}

\author{Redlef B.G. Braamhaar}
\affiliation{MESA+ Institute for Nanotechnology, University of Twente, 7500 AE Enschede, The Netherlands}

\author{Sara Marzban}
\altaffiliation{Current affiliation: PiCard Systems B.V., Toernooiveld 100,
6525EC, Nijmegen, The Netherlands.}
\affiliation{MESA+ Institute for Nanotechnology, University of Twente, 7500 AE Enschede, The Netherlands}

\author{Jens Eisert}
\affiliation{Dahlem Center for Complex Quantum Systems, Freie Universit{\"a}t Berlin, 14195 Berlin, Germany}
\affiliation{Helmholtz-Zentrum Berlin fur Materialien und Energie, 14109 Berlin, Germany}
\affiliation{Fraunhofer Heinrich Hertz Institute, 10587 Berlin, Germany}

\author{Jelmer J. Renema}
\affiliation{MESA+ Institute for Nanotechnology, University of Twente, 7500 AE Enschede, The Netherlands}

\affiliation{Department of Applied Physics and
Science Education, Eindhoven
University of Technology, P. O. Box 513, 5600 MB Eindhoven, The
Netherlands}

\affiliation{Department of Electrical Engineering, Eindhoven
University of Technology, P. O. Box 513, 5600 MB Eindhoven, The
Netherlands}

\author{Nathan  Walk}
\affiliation{Dahlem Center for Complex Quantum Systems, Freie Universit{\"a}t Berlin, 14195 Berlin, Germany}

\date{\today}
\begin{abstract}
Certification is important to guarantee the correct functioning of quantum devices. A key certification task is verifying that a device has produced a desired output state. In this work, we study this task in the context of photonic platforms, where single photons are propagated through linear optical 
interferometers to create large, entangled resource states for metrology, communication, quantum advantage demonstrations and for so-called linear optical quantum computing (LOQC). This setting derives its computational power from  the indistinguishability of the photons, i.e.,  their relative overlap. Therefore, standard fidelity witnesses developed for distinguishable particles (including qubits) do not apply directly, because they merely certify the closeness to some fixed target state. We introduce a measure of fidelity suitable for this setting and show several different ways to witness it, based on earlier proposals for measuring genuine multi-photon indistinguishability. We argue that a witness based upon the discrete Fourier transform is an optimal choice. We experimentally implement this witness and certify the fidelity of several multi-photon states.
\end{abstract}

\maketitle

\section{Introduction}

A crucial ingredient for developing useful, large-scale quantum devices is certifying that they function  as intended. {After all, we cannot hope to have useful large-scale quantum devices if the development is not accompanied by concomitant certification and benchmarking procedures
\cite{CertificationReview,PRXQuantum.2.010201,PropertyTesting}.}
Simple as this task may appear, {putting it into 
practice can be a 
challenging endeavor; for large-scale quantum system, learning the full process or 
quantum state 
in a variant of a tomographic procedure is infeasible, simply for the exponential sample complexity accompanying this task. 

Fortunately, in many meaningful situations and prospective applications of quantum technology, such full tomographic 
knowledge is not required.}
{Oftentimes -- in fact, this is more the rule rather than the exception --
one is merely interested in providing a meaningful answer to how close a given experimental preparation is to an anticipated one. In most
practically relevant}
cases, {one already has a pretty good}
idea of what the system is doing, i.e., what state it should prepare, and merely wants to check whether and to what extent that is indeed the case. 
{So while} traditional full state tomography remains impractical as quantum devices scale up, one can {substantially} reduce the computational resources required for certification by making some assumptions about the system in question or easing the requirements on the information content of the certification method. 

{The mindset of \emph{property testing} captures the idea that 
given a large object that either has a certain property or is somehow `far' from having that property, a tester should efficiently distinguish between these two 
cases. For state preparation tasks, the \emph{fidelity} of state preparations
is a meaningful quantity to capture closeness, and has led to the development of \emph{fidelity witnesses}, which are experimental procedures to lower bound the fidelity to a desired target state \cite{Leandro,Gluza:2018}.}

{While the basics of this are well understood and good reviews are available \cite{PropertyTesting,PRXQuantum.2.010201,CertificationReview}, for certain \emph{photonic platforms} -- one of the promising avenues for realizing scalable quantum computers} is via so-called \emph{linear optical quantum computing} (LOQC) \cite{KLM,Oneway,Browne:2005,Probabilistic2,Bartolucci:2023,Maring:2024,PsiQ:2025} or sub-universal quantum advantage devices \cite{aaronson_computational_2010,Wang:2019,Renema:2021sampling} {-- this task is marred by technical and conceptual difficulties.} These platforms proceed by encoding information in a particular degree of freedom of single photons (a common example is spatial, `dual-rail' encoding) that are then injected into linear optical interferometers that manipulate that degree of freedom or mode, commonly referred to as an \emph{external} mode. Complex multi-photon interference is then engineered to prepare a desired state which is subsequently measured by photon counting detectors. However, ideal bosonic interference will only occur when the photons involved are perfectly \emph{indistinguishable} in all other degrees of freedom, even those that the experimentalist cannot directly resolve or manipulate, which are typically referred to as \emph{internal}. Naturally, perfect indistinguishability is unachievable in practice and substantial work has been carried out showing the significant implications of distinguishability errors for the hardness of quantum sampling tasks \cite{Rohde:2012,Shchesnovich:2014,Rohde:2015,Shchesnovich:2015,renema2018efficient,Moylett:2018,Renema:2018,Moylett:2020,Shchesnovic:2021} and the reliability of fault-tolerant LOQC 
\cite{JelmerPhotonDistillation,MindTheGaps,saied:2025,Somhorst:2025,Bartolucci:2023}.

More broadly, in this context, there is a degree to which standard notions of state certification and tomography no longer make operational sense. When we write down a `perfect' implementation of some interference experiment, we typically specify an input state vector of $n$ identical photons to be injected into a particularly interferometer to prepare some target state. However, this description only contains \emph{relative} information, in that we have specified that the photons are indistinguishable from one another but have said nothing about the actual structure of the mode of which they are an excitation. In this sense, one might object that we have failed to fully specify a target state with respect to which we could evaluate a fidelity! More practically speaking, the fact that standard detectors are typically only capable of resolving external modes also creates difficulties for the would-be certifier wishing to use only equipment from the natural LOQC toolbox.


Although several proposals tailored to photonic systems have appeared, they all fall into one of two categories, neither of which serves to easily address these issues. The first are those designed with continuous variable systems in mind that rely on the availability of homodyne  or heterodyne detection to achieve fidelity witnesses \cite{Leandro,Chabaud:2020,BosonSamplingVerification,Upreti:2024} or even state \cite{Bittel:2024,Mele:2025,Bittel:2025} or process \cite{Fanizza:2025} tomography. However, these protocols are for certifying the fidelity to a fixed internal state and so are conceptually unsatisfactory for the reasons given above. Perhaps more importantly, although pulsed homodyne detection can be achieved in some hybrid implementations \cite{Xanadu:2025,Larsen:2025}, it remains highly challenging and is generally unavailable in most LOQC implementations. The second class of proposals utilize only standard LOQC components for tomography \cite{RahimiKeshari:2013,LaingOBrien,Banchi:2018,Hoch:2023} or to implement modern techniques such as randomized benchmarking \cite{Wilkens:2024,Arienzo:2025} and shadow tomography \cite{Thomas:2025} but essentially focus only on verifying that the interferometer correctly manipulates the external degrees of freedom, whilst assuming the input photons are perfectly indistinguishable. 

However, thanks to the enormous amount of work on the physics of partial photon distinguishability \cite{hong_measurement_1987,Ou:2006,Rohde:2007if,Tichy:2010,Tichy:2012gt,Tan:2013,Tichy:2014,Tillmann:2015,Rohde:2015,Shchesnovich:2015,Tichy:2015,Dittel:PRA,Dittel:PRL,Dittel:PRX,Menssen:2017,Stanisic:2018,Jones:2020,Jones:2023,Shchesnovich:2018triad,annoni2025} it has now been clarified how the observed level of bosonic interference is determined entirely by the permutation symmetries of the detected photons. It has been shown that a perfect LO interference experiment is realized not only by $n$ photons with identical internal degrees of freedom, but by any state from an equivalence class of $n$-photon states that are symmetric under all possible particle permutations. From the perspective of an observer who can only carry out LOQC experiments it makes no difference which precise state in the class is prepared. Similarly, from the perspective of a certifier, it is only operationally necessary to bound the fidelity with \emph{any} state in the target equivalence class. Our first contribution is therefore to recast the task of certification in the linear optics context by defining an equivalence class of target states, $\mathcal{C}_\mathrm{LO}$, as all states made by processing a fiducial state of $n$ indistinguishable single photons and $m-n$ vacuum modes through a given $m$-mode interferometer, $U$. We then define an LOQC version of the fidelity (trace distance) as the standard quantity suitably maximized (minimized) over that class of states that are effectively equivalent to a specific target state,
reminiscent of other metrics, such as the LOCC trace distance.

Defining the LOQC fidelity in this way provides a quantity that is not only operationally meaningful, but experimentally accessible using well chosen, trusted interferometers. Here, trust should be understood as the property that the interferometer used for certification can be safely assumed to act upon the external degrees of freedom as programmed while leaving the internal degrees of freedom undisturbed. 

Intuitively, we would expect a witness to require verification of both the external degrees of freedom manipulated by the interferometer and the permutation symmetry properties of all input photons. Given an interferometer that reliably manipulates external modes, the first property could be checked by applying the inverse of the unitary defining the target state (i.e.,  $U^\dag$) and monitoring how often the original arrangement of single-photons and vacuum states is observed. We will refer to this as measuring the \emph{photon reversibility}. To ascertain the photon indistinguishability we could perform an interference experiment with distinct behavior for (in)distinguishable photons. Here, we would also require that the certification interferometer leaves the internal modes untouched, so that observed interference behaviour can be safely attributed to the experiment being certified. Given these requirements, the observed interference statistics can be used to bound the projection in the $n$-photon symmetric subspace, and we show in Thm.~\ref{thm_sym} that the LOQC fidelity can be meaningfully bounded by separate measurements of photon indistinguishability and reversibility.

Indeed, a protocol along these lines was previously proposed \cite{Somhorst:2023} that applied the discrete Fourier transform and its well studied \emph{zero-transmission} properties \cite{Tichy:2010,Tichy:2012gt,Tichy:2014,Dittel:PRL,Dittel:PRA} to bound the multi-photon indistinguishability. We go beyond \cite{Somhorst:2023} in several ways, not only by defining and motivating the LOQC fidelity, which is not made explicit in Ref.\ \cite{Somhorst:2023}, but also by obtaining concrete performance bounds for arbitrary $n$. In Thm.~\ref{thm_wit} we provide a general method to transform from an indistinguishability measure (when applicable) into a fidelity witness and then consider three other proposals for quantifying indistinguishability, based upon bunching statistics of super-posed Hong-Ou-Mandel experiments \cite{Brod:2019cf}, cyclic interferometry \cite{Pont:2022} and a two-mode correlation measurement \cite{vanderMeer:2021}. Lastly, in Thm.~\ref{thm_exp} we derive a tighter bound from a witness that is more economical in terms of the optical depth of the trusted interferometer at the price of introducing stronger assumptions on the input states.

The properties of the fidelity witnesses are primarily determined by their constituent indistinguishability measures. Thus, we carry out a theoretical and numerical comparison in terms of tightness, assumptions required to apply the indistinguishability witnesses, partial device-independence with respect to the trusted interferometers and sample complexity. These comparisons are then validated by an experimental implementation of all four methods using spontaneous parametric down-conversion sources and a programmable integrated photonic device. 

While the cyclic and Fourier methods were already known to provide tight bounds, we show the others only provide good approximations in the limit of high photon indistinguishability. The assumptions used to derive the various witnesses, which correspond to restricting to certain kinds of incoherent distinguishability errors, can be unified via recent work on the so-called \emph{partition representations} \cite{annoni2025}. All witnesses can be applied if the input states possess a \emph{positive} partition representation when projected on the correct photon-number subspace, which is a reasonable assumption in many implementations and a condition that can be enforced by applying randomized permtutation channels analogous to qubit twirling \cite{PhysRevLett.76.722,Knill:2004,Wallman:2016}. We provide explicit examples of how the cyclic method fails for non-partition states and how the bunching method fails for states with a negative partition representation. Moreover, we demonstrate that a negative partition representation can be observed by a relatively simple time-delay model of distinguishability error, thereby adding to the recent literature on the deviation between bosonic and `bunched' behavior \cite{Seron:2023,Pioge:2023,Rodari:2024,Geller:2025}. Interestingly, we show that the Fourier method is invariant under permutation twirling and hence can be used to bound the fidelity that would be observed for any input state subjected to permutation twirling. 

The two-mode correlator method was shown to be semi-\emph{device-independent} (DI) in the sense that any errors in the trusted interferometer act only to lower the certified indistinguishability \cite{vanderMeer:2021}. We show that this property does not hold for the cyclic witness but provide numerical evidence that the Fourier and bunching methods are semi-DI. In 
Ref.\ \cite{Somhorst:2023}, the relevant detection probabilities are only explicitly computed for the case of three input photons so that the sample-complexity of the scheme was not clear. Here, for a slightly more restricted class of distinguishability noise (the so-called ``orthogonal bad bit" model \cite{Sparrow:2017}), 
we show the Fourier witness has an $\mathcal{O}(1)$ sample complexity and we conjecture this holds generally. This compares favourably to $\mathcal{O}(n^2)$ for the bunching witness, $\mathcal{O}(n^4)$ for the two-mode correlator and $\mathcal{O}(2^{2n})$ for the cyclic method.

Finally, based on these comparisons we conclude the Fourier indistinguishability witness enjoys the best combination of features and experimentally implement the corresponding fidelity witness derived in Thm.~\ref{thm_exp} to certify 
a $3$-photon, $4$-mode state preparation device for a selection of target states corresponding to Haar random interferometers.

This work is organized as follows: in Sec.~\ref{background} we present some necessary background and notation before defining the equivalence class of target states and the LOQC fidelity in Sec~\ref{sec:LOQC_fid}. In Sec.~\ref{witness} we derive fidelity witnesses under a different tradeoffs between assumptions on the input state and overall practicality. In Sec.~\ref{comparison} we present a detailed theoretical and numerical comparison of the four witnesses and finally in Sec.~\ref{experiment} we describe the experimental implementation and results of the four indistinguishability witnesses as well as the full fidelity witness utilising the Fourier method.

\section{Certification in linear optics \label{theory}}

This section establishes the background theoretical concepts and notation underlying our certification results.
In Sec.~II~A we review the description of partial photon indistinguishability in linear-optical networks and introduce the generalized indistinguishabilities \cite{Shchesnovich:2015}, defined as expectation values of permutation operators, which completely characterize LOQC output statistics. In Sec.~II~B we formalize the notion of LOQC equivalence-class of target states, $\mathcal{C}_{LO}$, and introduce operational fidelity and distance measures appropriate for this setting.
Finally, in Sec.~II~C we derive a general lower bound on the LOQC fidelity for various levels of assumptions (particularly the partition representation) and practicality. These bounds are formulated in terms of two experimentally accessible quantities: the probability observing the correct external-mode photon occupation, and a bound on the weight of the fully indistinguishable sector of the input state. 

\subsection{Partial indistinguishability theory \label{background}}
Consider an arbitrary Fock basis state vector with a total of $n$ photons over $m$-modes \eqn{\ket{\mathbf{n}}:=\ket{n_1}\ket{n_2}\cdots\ket{n_m} = \frac{1}{\sqrt{\mu(\mathbf{n})}} \prod_{i=1}^m \bk{\adag_{i}}^{n_i}\ket{0} \label{indist} } 
where $\mathbf{n} = [n_1,n_2, \dots, n_m]$ is called a \emph{mode occupation} list, $\sum_{i=1}^m n_i = n$ and  $\mu(\mathbf{n}) = \prod_{i=1}^m n_i!$. Of particular interest to us will be states where the first $n$ modes are occupied by a single photon and the remaining ($m$-$n$) modes are in a vacuum state, which we denote as $\ket{\mathbf{1}_{n,m}}$ (we further define the special case $\ket{\mathbf{1}_{n}}:=\ket{\mathbf{1}_{n,n}}$). Such states describe the input to a \textsc{BosonSampling} problem and typically serve as an ideal initial state for universal LOQC and other applications.  Alternatively, we could specify this state by a unique corresponding vector $\mathbf{d_n}$ of length $n$ (called a \emph{mode assignment} list) where $\mathbf{d_n}(i)$ specifies in which mode the $i^\mathrm{{th}}$ photon lies with which 
the state vector could equivalently be written
\eqn{\ket{\mathbf{n}} = \frac{1}{\sqrt{\mu(\mathbf{n})}} \prod_{i=1}^n \adag_{\mathbf{d_n}(i)} \ket{0}.}
The remainder of an LOQC computation consists of propagation through a linear optical interferometer which can be represented by an $m\times m$ unitary matrix, $U$, which transforms the creation operators as 
\eqn{\adag_i \mapsto  \sum_{k=1}^m U_{i,k}\hat{b}^\dag_k \label{Ua}.}
After propagation, the output state is measured by photon number resolving detectors at each output port, which again project onto Fock basis states
\eqn{\ket{\mathbf{s}} = \frac{1}{\sqrt{\mu(\mathbf{s})}} \prod_{i=1}^m \bk{\hat{b}^\dag_{i}}^{s_i}\ket{0}.}
The probability of observing at output $\mathbf{s}$ given an input $\mathbf{n}$ propagated through $U$ is then
\eqn{p(\mathbf{s}|\mathbf{n})_U = \left |\bra{\mathbf{s}}\hat{V}_U\ket{\mathbf{n}}\right|^2 =\frac{|\mathrm{perm}(M)|^2}{\mu(\mathbf{s})\mu(\mathbf{n})},\label{pn}} 
where $M$ is constructed from the elements 
of $U$ via
\eqn{M_{i,j} = U_{\mathbf{d_n}(i),\mathbf{d_s}(j)} \label{M}}
and $\hat{V}_U$ is the representation of $U$ acting in the appropriate Hilbert space for $\ket{\mathbf{n}}$. Recalling that $U_{i,j}$ is the amplitude of a photon transitioning from input mode $i$ to output mode $j$ we can interpret each $M_{i,j}$ as the amplitude of a collection of `paths' by which the photons in an initial mode assignment, $\mathbf{d_n}$, collectively transition so as to be found in a final mode assignment $\mathbf{d_s}$. However, for bosons that are entirely indistinguishable one cannot just consider a single $M_{i,j}$ (as this would be effectively assigning physical meaning to the mathematical label ascribed to individual photons). 
Instead, we must coherently add all of the amplitudes corresponding to all permutations of the $n$ photons that give rise to the same $\mathbf{d_s}$. This is achieved via the matrix permanent $\mathrm{perm}(M) = \sum_{\sigma \in \mathcal{S}_n} \prod_{i=1}^mM_{i,\sigma(i)}$.

In general, it is not possible to create photons that are perfectly mode-matched in all degrees of freedom and hence truly indistinguishable. A more realistic description must include so-called \emph{internal} degrees of freedom (modes) which cannot be directly manipulated or resolved by the experimentalist. A typical case in integrated photonic experiments is that the external degrees of freedom are the spatial modes with all other degrees of freedom 
(e.g.,  polarisation, orbital angular momenutum, frequency) being internal. Mathematically, this can be represented by writing creation operators $\hat{a}^\dag_{i,\kappa}$ where the Greek subscript $\kappa$ represents an internal mode. Without loss of generality, all internal degrees of freedom can be `pooled' and described by a single internal Hilbert space consisting of at most $n$ elements (to allow the distinguishability behavior to range from all photons being indistinguishable to completely distinguishable and all intermediate situations). An arbitrary state (mixed or entangled) of exactly $n$-photons with a mode occupation list {\bf n} can thus be written as
\eqn{\rho_\mathbf{n} \ee\nn \frac{1}{\mu(\mathbf{n})} \sum_{\boldsymbol \kappa \boldsymbol \kappa'} G(\boldsymbol \kappa,\boldsymbol \kappa') \\
&\times& \prod_{j=1}^n \adag_{\mathbf{d_n}(j),\kappa_j} \ket{0}\bra{0} \prod_{j=1}^n \hat{a}_{\mathbf{d_n}(j),\kappa'_j} \label{rho_n}}
where the sums run over all lists $\boldsymbol \kappa = [\kappa_1,\kappa_2, \dots, \kappa_n]$ with $1\leq \kappa_i\leq n$ of internal basis states and $G(\boldsymbol \kappa,\boldsymbol \kappa')$ are the corresponding amplitudes. 

The fact that the experimentally controllable unitaries act only upon the external degrees of freedom means the transformation under $U$ now reads
\eqn{\hat{a}^\dag_{i,\kappa} \mapsto \sum_{j=1}^m U_{i,j} \hat{b}^\dag_{j,\kappa}.}
Another notable property of LOQC experiments is that measurement are typically carried out with detectors which only resolve external degrees of freedom whilst summing over all internal degrees of freedom. The probability of observing a particular pattern on such detectors is then given by
\eqn{p(\mathbf{s}|\rho_\mathbf{n}) = \mathrm{tr} \left (\rho_\mathbf{n} \Pi(\mathbf{s}) \right)}
where the POVM elements are given by
\eqn{\Pi(\mathbf{s}) = \frac{1}{\mu(\mathbf{s})} \sum_{\boldsymbol \kappa} \prod_{i=1}^n \hat{b}^\dag_{\mathbf{d_s}(i),\boldsymbol \kappa} \ket{0}\bra{0} \prod_{i=1}^n \hat{b}_{\mathbf{d_s}(i),\boldsymbol \kappa}  .}
Evaluating these probabilities is challenging in general. However, using the path-based formalism of Shchesnovich \cite{Shchesnovich:2015} one can derive the expression (see Appendix~\ref{partial_app} for more details)
\eqn{P(\mathbf{s}|\rho_\mathbf{n}) = \frac{1}{\mu(\mathbf{n})\mu(\mathbf{s})} \sum_{\sigma \in S_n} \EV{J_\sigma}_{\rho_\mathbf{n}} \mathrm{perm} \bk{M * M^*_{\sigma,\mathbb{I}},} \label{Shchesnovich1}} 
where $*$ denotes element-wise multiplication, $M^*_{\sigma,\mathbb{I}}$ denotes the matrix defined by (\ref{M}) with the columns permuted according to $\sigma$ and $J_\sigma$ is a the permutation operator over the Hilbert spaces of the internal degrees of freedom of each photon, i.e., 
\eqn{J_\sigma  \prod_{j=1}^n  \adag_{\mathbf{d_n}(j),\kappa_j} =  \prod_{j=1}^n \adag_{\mathbf{d_n}(j),\kappa_{\sigma(j)}}.\label{Jsig}} For input state vectors of the form $\ket{\mathbf{n}}$ where all photons are perfectly indistinguishable we have that $\EV{J_\sigma}_{\rho_\mathbf{n}} = \mathrm{tr}(J_\sigma \rho_\mathbf{n})=1$ $\forall \sigma$. In fact, this condition holds for any permutationally invariant state which is any state invariant under the projector onto the completely symmetric subspace defined as 
\eqn{\Pi_\mathrm{sym} = \frac{1}{n!} \sum_{\sigma \in \mathcal{S}_n} J_\sigma \label{psym}.}
For all such states, the expression in 
Eq.~(\ref{Shchesnovich1}) reduces to (\ref{pn}) as expected.

The expectation values of the $n!$ permutation operators (referred to as \emph{generalized indistinguisabilities}) provide a complete specification of the behavior of arbitrary, partially distinguishable states of definite photon number of the form of Eq.~(\ref{rho_n}). In general, the permutation operator may not be physically 
easily realizable. However, in the special case of states that we denote as \emph{single-occupation states}, $\rho_{\mathbf{1}_n}$, where $m=n$ and all modes are occupied by exactly one photon ($\mathbf{n} = \mathbf{1}_n$ and hence $\mathbf{d_n} = [1,2, \dots, n]$) then the permutation of the external modes
\eqn{P_\sigma \prod_{j=1}^n  \adag_{\mathbf{d_n}(j),\kappa_j} =  \prod_{j=1}^n \adag_{\mathbf{d_n}(\sigma(j)),\kappa_{j}} \label{Psig}} results in the same output as $J_{\sigma^{-1}}$ (see Appendix~ \ref{partial_app}). Importantly, external mode permutations can be realized straightforwardly by simple switching operations. By making this observation \cite{annoni2025} were able to identify a family of LOQC experiments (generalisations of the cyclic interferometer in
Ref.\ \cite{Pont:2022}) that directly measure the $n!$ $\EV{P_\sigma}$ values which would enable a kind of `LOQC' tomography for this class of single photon states.

In our certification protocols we will frequently be concerned with projection onto the single-occupation external-mode subspace with one photon in each of the first $n$ modes and vacuum in the remaining $m-n$ modes, irrespective of the internal degrees of freedom. This projector can be written as
\begin{equation}
P_{1_{n,m}}
:=
\sum_{\boldsymbol{\kappa}\in\mathcal{K}^n}
\left(
\prod_{j=1}^{n}
a_{j,\kappa_j}^\dagger \,
|0\rangle\!\langle 0|\,
a_{j,\kappa_j}
\right)
\;\otimes\;
\bigotimes_{j=n+1}^{m}
|0\rangle\!\langle 0|_j ,
\label{eq:P_1nm_def}
\end{equation}
where $\kappa_j$ labels an orthonormal basis of the internal single-photon Hilbert space and the sum runs over all internal configurations.

\subsection{Certification metrics for LOQC \label{sec:LOQC_fid}}
Perhaps the conceptually simplest quantum certification task is to ascertain if the output state of an experiment or device is `close' to a hypothesized \emph{target state}. A commonly used tool in this context are \emph{witnesses} - measurement protocols that provide a certificate that the experimentally generated state $\rho$ achieves at least some threshold closeness with a target state except with some failure probability. The most common metrics used for this purpose are the quantum fidelity \cite{Leandro,Gluza:2018}
\eqn{F(\rho,\rho_t): = \mathrm{tr} \left[ \bk{\sqrt{\rho_t}\rho\sqrt{\rho_t}}^{1/2} \right]^2}
or the trace distance
\eqn{D(\rho,\rho_t) := \frac{1}{2} 
\mathrm{tr}(|\rho-\rho_t|),}
for which the Fuchs–van de Graaf inequalities hold,
\eqn{1-\sqrt{F(\rho,\rho_t)}\leq D(\rho,\rho_t) \leq \sqrt{1-F(\rho,\rho_t)}.}
Following Ref.\ \cite{Leandro}, we can formalize these notions as follows.

\begin{definition}[Certification protocols]
 Let $\rho_t$ be a target state, $0<T<1$ a threshold, and $\varepsilon>0$ a failure probability. A certification test, $\mathcal{T}$, which measures $k$ copies of an experimentally prepared state $\rho$ and either `accepts' or 'rejects' is called a, $(T,\varepsilon)$-certification test for the fidelity (resp. trace distance) if, with probability at least $1-\varepsilon$, it both rejects every state for which $F(\rho,\rho_t)<T$ (resp. $D(\rho,\rho_t)>T)$ and accepts if $\rho = \rho_t$. 
 \end{definition}
Such tests are referred to as witnesses because an acceptance witnesses a desirable threshold quality (e.g., $F(\rho,\rho_t)>T$ except with probability $\varepsilon$) but a rejection tells us essentially nothing about the quality of our state or the potential error processes that are occurring.

As outlined earlier, the quality of a photonic state for the purposes of LOQC is determined by its symmetry properties, which are determined by the relative arrangement of the internal degrees of freedom. As a consequence, certifying the fidelity to a specific state with precisely specified internal degrees of freedom is not strictly necessary. Like many aspects of quantum optics this issue can be effectively understood in the two-photon case by the canonical \emph{Hong-Ou-Mandel} (HOM) interference experiment \cite{hong_measurement_1987}. When two indistinguishable photons interfere on a balanced beamsplitter, the destructive interference between the paths leading to coincident detection results in only bunched outputs being observed. This is true even if the detectors at each beamsplitter output port cannot distinguish any internal degrees of freedom. If the photons are rendered completely distinguishable by some internal mode then suppression of coincidences vanishes and for partially distinguishable photons a corresponding `dip' in the coincidence probability is observed. A key insight is that output statistics are determined not by the internal degrees of freedom in an absolute sense but solely by the photon distinguishability which is a function of how one photon's internal modes are arranged relative to the other. For example, it is immediately apparent that perfect quantum interference will be observed as long as both photons are in the same internal mode, regardless of what that mode is. Moreover, as illustrated in Fig.~(\ref{fig:C_LO}), perfect interference will also be observed for a mixture of all of the photons being simultaneously in any two orthogonal internal modes (e.g., all `blue' or all `green') or in an even superposition. The reason for this behavior is that all of these states are completely symmetric under permutation - the hallmark of bosonic behavior. The observation of photonic interference due to symmetric, anti-symmetric and even anyonic exchange symmetries have been experimentally demonstrated in two-photon \cite{Exter:2012} and multi-photon \cite{Kumar:2025} experiments.

\begin{figure}[htbp]
    
    \includegraphics[width=1.05\linewidth]{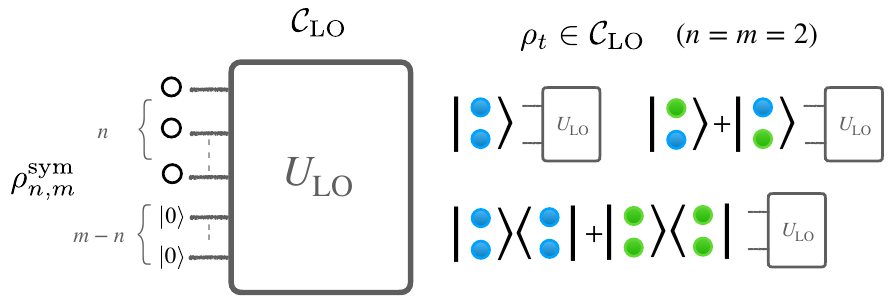}
    \caption{{\bf Certification target class for linear optics.} The class $\mathcal{C}_{\mathrm{LO}}$ contains all of the $m$-mode states consisting of the first $n$ modes containing a completely permutationally invariant state with exactly 1 photon in each mode and the remaining $m-n$ modes in the vacuum state (we denote such a state $\rho^\mathrm{sym}_{n,m}$) that is evolved through a linear optical unitary $U_\mathrm{LO}$. 
    We provide some examples of states in $\mathcal{C}_{\mathrm{LO}}$ for the case ($n=m=2$) to 
    highlight that as well as the canonical case of two identical (blue) photons entering the unitary we may also have a statistical mixture all photons being in orthogonal internal modes (all blue or all green) or even an even superposition of each photon being in one internal mode or the other. As explained in the main text, so long as the input state is symmetric under all permutations then all these states are functionally equivalent in the sense that they realize a `perfect' LOQC experiment.}
    \label{fig:C_LO}
\end{figure}

In fact, we can see directly from Eq.~(\ref{Shchesnovich1}) that this behavior generalizes to arbitrary interference experiments. The internal state of the photons only affects the outcome of LOQC circuits through the generalized indistinguishabilities, so all permutationally invariant states with a mode occupation list $\mathbf{n}$, injected into a linear optical unitary and then measured by photon-number detectors that are insensitive to internal degrees of freedom, will result in the same statistics as the input state vector $\ket{\mathbf{n}}$. This notion can also be formalized.

\begin{definition}[LOQC equivalence]
Any two states $\rho_1$ and $\rho_2$ are said to be LOQC equivalent, $\rho_1\sim\rho_2$, iff they produce identical probability distributions for all LOQC experiments such that
\eqn{p(\mathbf{s}|\rho_1)_U = p(\mathbf{s}|\rho_2)_U \hspace{1mm} \forall \mathbf{s},U  \Leftrightarrow \EV{J_\sigma}_{\rho_1} = \EV{J_\sigma}_{\rho_2} \forall \sigma.}
\end{definition}
The equivalence between these two conditions follows immediately from (\ref{Shchesnovich1}). Based on these considerations we must slightly amend our notions of certification to the LOQC setting. The motivation here is to extend the target certification class, $\mathcal{C}_\mathrm{LO}$, to include all possible states that could be regarded as operationally perfect LOQC state preparations. This can be achieved by defining the class as all states that are LOQC equivalent to the canonical preparation of the state $\ket{\mathbf{1}_n}$, or equivalently as any permutationally invariant state with the correct $\mathbf{d}_n$, propagated with vacuum modes through the appropriate $U$.

\begin{definition}[Class of LOQC target states $\mathcal{C}_\mathrm{LO}$]
 Let $U \in U(m)$ be an $m$-mode linear optical unitary (and $\hat{V}_U$ its metaplectic representation) we define the class of LOQC target states 
 \eqn{\mathcal{C}_\mathrm{LO}\ee \left\{\rho_t :\rho_t \sim \hat{V}_U\ket{\mathbf{1}_{n,m}}\right \}}
as the set of all states that are LOQC equivalent to $n$ indistinguishable photons appended with $(m-n)$ vacuum modes and then evolved through $U$. 

Equivalently, defining $\rho_{\mathbf{1}_n}$ as an $n$-mode state with exactly one photon per mode, $\mathcal{C}_{\mathbf{1}_n}^\mathrm{sym}$ be the set of all such states that are completely symmetric under all permutation and hence satisfying
\eqn{\mathcal{C}_{\mathbf{1}_n}^\mathrm{sym} = \left\{\rho_{\mathbf{1}_n} : \Pi_\mathrm{sym} \rho_{\mathbf{1}_n} \Pi_\mathrm{sym} = \rho_{\mathbf{1}_n} \Leftrightarrow \EV{P_\sigma} = 1 \hs \forall \sigma\right\}\label{c1sym}}
and using the shorthand $\rho_{\mathbf{1}_{n,m}} = \rho_{\mathbf{1}_{n}}\otimes\ketbra{0}{0}^{\otimes(m-n)}$ we could define this target class
\eqn{\mathcal{C}_\mathrm{LO}= \left\{\rho_t = V_U \rho_{\mathbf{1}_{n,m}}V_U^\dag : \rho_{\mathbf{1}_n}\in \mathcal{C}_{\mathbf{1}_n}^\mathrm{sym}\right \}}
as the set of all permutationally invariant $n$-mode single photon states appended with $(m-n)$ vacuum modes and then evolved through $U$. 

We similarly define a slightly more general class where the target states are LOQC equivalent to arbitrary permutation invariant states of fixed photon number inputs. For a state $\rho_\mathbf{m}$ be an $m$-mode state of the form Eq.~(\ref{rho_n}) let $\rho_{\mathbf{n},m} = \rho_{\mathbf{n}}\otimes \ket{0}\bra{0}^{(m-n)}$ and then define
\eqn{\mathcal{C}_{\mathbf{n}}^\mathrm{sym} = \left\{\rho_{\mathbf{n}} : \Pi_\mathrm{sym} \rho_{\mathbf{n}} \Pi_\mathrm{sym} = \rho_{\mathbf{n}} \Leftrightarrow \EV{P_\sigma} = 1 \hs \forall \sigma\right\}\label{cnsym}} and hence
\eqn{\mathcal{C}_\mathrm{LO}^\mathrm{Fock}= \left\{\rho_t = V_U \rho_{\mathbf{m}}V_U^\dag : \rho_{\mathbf{m}}\in \mathcal{C}_{\mathbf{m}}^\mathrm{sym}\right \}}
or equivalently
 \eqn{\mathcal{C}_\mathrm{LO}^\mathrm{Fock}\ee \left\{\rho_t :\rho_t \sim \hat{V}_U\ket{\mathbf{m}}\right \}.}
\end{definition}

Equipped only with photodetectors that cannot resolve internal degrees of freedom we could not hope to resolve any one particular state in $\mathcal{C}_\mathrm{LO}$ and certify its fidelity with an experimentally prepared state. However, such certification is also operationally unnecessary when assessing the quality of an LOQC device. In the context of LOQC, we are only really interested in the correctness of the output photon counting statistics. As such, for any target state $\rho_t\in\mathcal{C}_\mathrm{LO}$ an appropriate metric should assign an ideal score to a device that produces any state $\rho\sim\rho_t$.

These considerations can be seen as analogous to other works that consider distinguishability and certification in other well-motivated, operationally restricted contexts such as LOCC operations \cite{Bennett:1999vm,Walgate:2002,Mathews:2009,lancien:2013} or efficient (poly-depth) circuits \cite{Ji_pseudorandom,AaronsonPseudoentanglement,arnon:2023,Leone:2025}. The 
common thread is that, when attempting to certify a resource for our own use, we may only be capable of utilising and distinguishing resources in some restricted setting, and the question of real interest is not whether we have prepared a specific target state but rather if we have prepared \emph{any} state that is operationally `just as good' in the relevant setting. Consequently, although we may sometimes be interested in ultimate limits on distinguishability (e.g.,  when considering an all-powerful opponent in some cryptographic protocol), it may sometimes be more useful to consider equivalence classes of states that are indistinguishable under certain well-motivated restrictions. This has led to various flavours of restricted distinguishing metrics and motivates the following definitions tailored to the LOQC setting.

\begin{definition}[LOQC certification metrics]
\label{BigDefinition}
For any $\rho_t \in \mathcal{C}_\mathrm{LO}$ the effective LOQC fidelity is defined as
\eqn{F_\mathrm{LO}(\rho,\rho_t) \ee \max_{\tau\sim\rho_t} F(\rho,\tau) \label{fiddef1}}
where
$F(\rho,\rho_t)$ is the standard quantum fidelity. Similarly, the LOQC trace distance is defined as
\eqn{D_\mathrm{LO}(\rho,\rho_t) = \min_{\tau\sim\rho_t} D(\rho,\tau) \label{trdef1}}
where $D(\rho,\rho_t)$ is the trace distance. These quantities can equivalently be written as 
\eqn{F_\mathrm{LO}(\rho,\rho_t) =  \max_{\rho_{\mathbf{1}_n} \in \mathcal{C}_{\mathbf{1}_n}^\mathrm{sym}} F\bk{\rho,\hat{V}_U\rho_{\mathbf{1}_{n,m}}\hat{V}_U^\dag)}, \label{fiddef2}}
\eqn{D_\mathrm{LO}(\rho,\rho_t) =  \min_{\rho_{\mathbf{1}_n} \in \mathcal{C}_{\mathbf{1}_n}^\mathrm{sym}} D\bk{\rho,\hat{V}_U\rho_{\mathbf{1}_{n,m}}\hat{V}_U^\dag)}. \label{tfdef2}}
\end{definition}
Considering the optimization in the above definition we see that these quantities fulfill our goal of quantifying the distance to the closest state that is operationally equivalent to the target state for the purposes of implementing an LOQC circuit.

Lastly, we mention here another notion of the fidelity and trace distance that could also be seen as natural in the context of LOQC. It is well known that both quantum metrics are equal to the corresponding classical quantities defined on probability distributions obtained by measuring the two states to be distinguished with appropriately optimized measurements. We can immediately define an equivalent notion for observers restricted to only making LOQC measurements by defining 

\eqn{\tilde{F}_{\mathrm{LO}}(\rho,\sigma)  \ee \min_{\{ F_\mathbf{s}\}}\left(\sum_\mathbf{s}\sqrt{\mathrm{tr}\left (F_\mathbf{s} \rho\right )\mathrm{tr}\left (F_\mathbf{s} \sigma\right )}\right)^2 ,\\
\tilde{D}_{\mathrm{LO}}(\rho,\sigma)  \ee \max_{\{ F_\mathbf{s}\}} \frac{1}{2} \sum_\mathbf{s} \left | \mathrm{tr}\left (F_\mathbf{s} \rho\right ) - \mathrm{tr}\left (F_\mathbf{s} \sigma\right )\right|, \label{alt_cert}}
where we restrict the optimization to LOQC measurements of the form $F_{\mathbf{s}} = U_\mathrm{LO}\ket{\mathbf{s}}\bra{\mathbf{s}}U_\mathrm{LO}^\dag$ rather than optimizing over arbitrary measurements. Performing the optimization over LOQC measurements to obtain these quantities exactly is non-trivial, however, we note that that they are usefully bounded by the quantities we certify in this work, namely $\tilde{F}_{\mathrm{LO}}(\rho,\sigma) \geq F_{\mathrm{LO}}(\rho,\sigma)$ and $\tilde{D}_{\mathrm{LO}}(\rho,\sigma) \leq D_{\mathrm{LO}}(\rho,\sigma)$ (see App.~\ref{alt_defs}.

\subsection{Certification setup and assumptions \label{witness}}
Our goal now is to obtain measurement protocols that will allow us to meaningfully lower bound the LOQC fidelity, hence obtaining rigorous ($T,\varepsilon$)-certificates for the LOQC regime. We will also explain in detail what equipment and assumptions are required. To that end, we first provide a certification protocol with slightly abstracted measurements and then an exact specification of assumptions and experimental procedures that can achieve LOQC certification with present-day technology.

 \begin{theorem}[LOQC fidelity witness]
 \label{thm_sym}Given constants $\varepsilon,\delta_1,\delta_2 >0$ a target state $\rho_t = \rho_{\mathbf{1}_{n,m}}^U \in \mathcal{C}_\mathrm{LO}$ and $k$ copies of unknown state $\rho$, let $\mathcal{C}$ be a certification protocol that can measure the observables: 
 \eqn{p_1 = \tr(\tilde{\rho} P_{\mathbf{1}_{n,m}})}
 which we denote the \emph{photon reversiblity} with $k_1$ samples to obtain an estimate $\bar{p}_1$ where $\tilde{\rho} = \hat{V}_U^\dag \rho \hat{V}_U$, is the unknown experimental state `reversed' under the target interferometer and $P_{\mathbf{1}_{n,m}}$ is the projection operator onto $n$ single photons and $m-n$ vacuum modes defined in Eq.~(\ref{eq:P_1nm_def}); and 
 \eqn{p_2 = \tr(\tilde{\rho} (\Pi_\mathrm{sym}^\mathbf{1}\otimes\mathbb{I}
 _{m-n})),} with $k_2$ samples to obtain an estimate $\bar{p}_2$, where
 \eqn{\Pi_\mathrm{sym}^\mathbf{1} = \frac{1}{n!} \sum_{\sigma \in \mathcal{S}_n} P_\sigma}
 is the mode permutation operator acting on the first $n$ modes. 
Then $\mathcal{C}$ is an ($T,\varepsilon$)-fidelity witness for the LOQC fidelity with
\eqn{T=  \bar{p}_1 + \bar{p}_2 - 1 - \delta_1 - \delta_2 \label{fidthresh}}
and a minimum number of copies $k=k_1+k_2$ where,
\eqn{k \geq \frac{\mathrm{ln}(2/\varepsilon)}{2} \bk{\frac{1}{\delta_1^2} + \frac{1}{\delta_2^2}},}

 \end{theorem}
 \emph{Proof}: We begin by proving a slightly more general result. Starting from the LOQC fidelity definition in Eq.~(\ref{fiddef2}) we can generalize the quantity to target states in $\mathcal{C}_\mathrm{LO}^\mathrm{Fock}$ to obtain,
\eqn{F_\mathrm{LO}^\mathrm{Fock}(\rho,\rho_t) \ee  \max_{\rho_{\mathbf{m}} \in \mathcal{C}_{\mathbf{m}}^\mathrm{sym}} F\bk{\rho,\hat{V}_U\rho_{\mathbf{m}}\hat{V}_U^\dag }\nn \\
\ee \max_{\rho_{\mathbf{m}} \in \mathcal{C}_{\mathbf{m}}^\mathrm{sym}} F\bk{\hat{V}_U^\dag\rho \hat{V}_U,\rho_{\mathbf{m}}}}
where we have utilized the unitary invariance of the fidelity in the second equality. Recalling that maximizing the fidelity over all states in a subspace is the same as projection onto that subspace, we can replace the maximization over the subspace, $\mathcal{C}_{\mathbf{1}_n}^\mathrm{sym}$, by its projection. This is 
given by $P_{\mathbf{n}}\Pi_\mathrm{sym}$ where $P_{\mathbf{n}}$ is the projection onto a photon number occupation specified by mode occupation list $\mathbf{n}$ (regardless of internal degrees of freedom) and $\Pi_\mathrm{sym}$ is the projection on the symmetric subspace of the $n$ photons. Note that since $\Pi_\mathrm{sym}$ acts only upon the internal modes and $P_{\mathbf{n}}$ only upon the external modes, these two projectors commute. Setting $\tilde{\rho} = \hat{V}_U^\dag\rho \hat{V}_U$, we then have 
\eqn{F_\mathrm{LO}^\mathrm{Fock}(\rho,\rho_t) = \tr \bk{\tilde{\rho} P_{\mathbf{m}}\bk{ \Pi_\mathrm{sym}\otimes\mathbb{I}_{(m-n)} }}.}
For any two outcome projector $P_1 + P_1^\perp = \mathbb{I}$ and any positive operator $P_2$ it holds that 
\begin{eqnarray}
\tr(\rho P_2) &=& \tr(\rho P_2(P_1 + P_1^\perp )) \nonumber\\
&=&  \tr(\rho P_2P_1)+\tr(\rho P_2P_1^\perp)\geq \tr(\rho P_2P_1)+\tr(\rho P_1^\perp)\nonumber\\ 
&=&\tr(\rho P_2P_1)+1-\tr(\rho P_1)
\end{eqnarray}
which implies $\tr(\rho P_2P_1) \geq \tr(\rho P_1)+\tr(\rho P_2)-1$ which in turn yields
\eqn{F_\mathrm{LO}^\mathrm{Fock}(\rho,\rho_t) \geq  \tr \bk{\tilde{\rho}P_{\mathbf{n}}} +\tr\bk{ \tilde{\rho} \Pi_\mathrm{sym}} - 1. \label{fidfock}}
This result forms the basis of a fidelity witness for all states $\rho_t \in \mathcal{C}_\mathrm{LO}^\mathrm{Fock}$ which we report here as it may be of future interest. As explained earlier, it is unclear how to measure the projection on the symmetric subspace for such states because of the difficulty of implementing the `photon' permutation operators $J_\sigma$. However, for special case of LOQC states we have that $P_{\mathbf{n}} = P_{\mathbf{1}_{n,m}}$ where $P_{\mathbf{1}_{n,m}}$ is the projection onto exactly one photon in the first $n$ modes, irrespective of internal degrees of freedom, and vacuum in the following $(m-n)$ modes. This yields
\eqn{F_\mathrm{LO}(\rho,\rho_t) \ee  \max_{\rho_{\mathbf{1}_n} \in \mathcal{C}_{\mathbf{1}_n}^\mathrm{sym}} F\bk{\rho,\hat{V}_U\rho_{\mathbf{1}_{n,m}}\hat{V}_U^\dag)} \nn \\
\ee \max_{\rho_{\mathbf{1}_n} \in \mathcal{C}_{\mathbf{1}_n}^\mathrm{sym}} F\bk{\hat{V}_U^\dag\rho\hat{V}_U,\rho_{\mathbf{1}_{n,m}})} \nn\\
\ee \tr \bk{\tilde{\rho} P_{\mathbf{1}_{n,m}}\bk{ \Pi_\mathrm{sym}\otimes\mathbb{I}_{(m-n)} }},}
where have again projected onto the appropriate subspace in the final line. 

Furthermore, given the $P_{\mathbf{1}_{n,m}}$ projection, we could now define the projector onto the symmetric subspace via mode permutation operators,  $\Pi_\mathrm{sym}^\mathbf{1}$. This is because, after the $P_{\mathbf{1}_{n,m}}$ projection, the state $\tilde{\rho}$ will by definition be in the special case of one photon in the first $n$ modes. Thus, we can bound the LOQC fidelity as
\eqn{F_\mathrm{LO}(\rho,\rho_t) \ee  \tr \bk{\tilde{\rho} P_{\mathbf{1}_{n,m}}\bk{ \Pi_\mathrm{sym}^\mathbf{1}\otimes\mathbb{I}_{(m-n)} }} \nn \\
&\geq &  \tr \bk{\tilde{\rho}P_{\mathbf{1}_{n,m}}} +\tr\bk{ \tilde{\rho} \bk{ \Pi_\mathrm{sym}^\mathbf{1}\otimes\mathbb{I}_{(m-n)} }} - 1 \nn \\
\ee p_1 + p_2 - 1.\label{fidbound}}
The final part of the proof simply consists of applying the Hoeffding's inequalities \cite{Hoeffding1963} which relate the observed frequencies, $\bar{p}$, of $k$ independent trials of a random variable to the true 
probability, $p$, via
\eqn{\mathrm{Pr}[|\bar{p}-p|>\delta]\leq \varepsilon:= \exp(-2 \delta^2 k).\label{hoeff}}
This result implies that, given 
\eqn{k_1 = \frac{\mathrm{ln}(2/\varepsilon)}{2\delta_1^2}} 
samples, estimating $\bar{p}_1$ we can say that with probability at least ($1-\varepsilon/2$) that
\eqn{p_1 \geq \bar{p}_1 - \delta_1}
and similarly for $p_2$. Putting this together with Eq.~(\ref{fidbound}) we have that with probability at least ($1-\varepsilon)$ that $F_\mathrm{LO}(\rho,\rho_t)>T$, with $T$ given by 
Eq.~(\ref{fidthresh}). Equivalently we could say that a protocol that accepts only upon achieving a threshold equal to the observed value of $\bar{p}_1 + \bar{p}_2$ will reject every state for which $F_\mathrm{LO}(\rho,\rho_t)<T$ with probability at least $(1-\varepsilon)$ thus completing the proof.
$\square$

A few comments about the reasonableness of measuring these quantities are in order. Firstly, it should be noted that given access to fully mode-resolving photodetectors, we would essentially obtain a fidelity witness immediately by simply measuring the first term in Eq.~(\ref{fidbound}). Given such detectors we could specify a precise mode (i.e., all internal degrees of freedom) for target state and avoid the necessity of the notion of an LOQC fidelity entirely albeit at the cost of potentially throwing away some states that would be functionally identical to our target states. Indeed, this is precisely how the witness of Ref.\ \cite{Leandro} has been obtained. There, the authors have defined their linear optical target class as $\{\hat{V}_U\ket{\mathbf{1}_{n,m}} \}$ where a particular internal degrees of freedom for the single photons was specified and then bounded the standard fidelity as \eqn{F(\rho,\rho_t) = F\bk{\hat{V}_U^\dag\rho \hat{V}_U,\ketbra{\mathbf{1}_{n,m}}{\mathbf{1}_{n,m}}}} 
which could be immediately measured with fully mode-resolving photodetectors. Although such detectors are not available, the key insight of Ref.\ \cite{Leandro} was that, given the additional ingredient of homodyne or heterodyne detectors (which are fully mode resolving due to the fact that only light in the same mode as the local oscillator is amplified during homodyne detection) the necessary expectation values could be inferred from the homo-/heterodyne data. Moreover, they showed that the time reversal unitary can be performed virtually via data post-processing of the single mode measurements homodyne measurements. A similar property holds for the work based upon heterodyne measurements in Refs.\ \cite{Chabaud:2020,Upreti:2024,BosonSamplingVerification}.

However, homodyne detection is often infeasible or unavailable in LOQC setups due to a lack of the requisite detectors or local oscillators at appropriate wavelengths, particularly as pulsed detection is essential in this context. Even if available, it is arguably an undesireable overhead to introduce an entirely new detection scheme purely for certification that is not utilized at all in the actual quantum computing architecture. 

Alternatively, the first term in our lower bound can be directly measured in an LOQC experiment consisting of an interferometer which can be reliably programmed to manipulate the external degrees of freedom and physically implement the inverse interferometer, $U^\dag$, followed by external mode resolving photo-detection. We call this measurement the photon-reversibility. In terms of resources and assumptions, in one sense this is simply the standard scenario for tomography or certification where a trusted measurement apparatus is used to certify an untrusted state preparation procedure. However, in some certification or tomography protocols, it is possible to learn or certify complicated states made from complex unitaries where the trusted devices are only simple, single-mode measurements. Here, we are requiring a close-to-ideal version of an interferometer of at least the same size and complexity as the untrusted one which is part of the state preparation procedure we are attempting to certify. Nevertheless, certification remains well-motivated in this setting. A real-world certification scenario could involve a single, exhaustively characterized and calibrated interferometer (a process which can also be done with entirely classical light \cite{RahimiKeshari:2013}) that is then used to certify the functioning of a large number of untrusted state-preparation modules (photon  source + interferometer). This process can be made robust to residual calibration errors in the detector by considering the degree to which these errors propagate into the certification. By this protocol, we can effectively bootstrap our trust from one device to many. 

Considering the second term in (\ref{fidbound}), it turns out that this can also be measured employing only the detectors and trusted interferometer already assumed to carry out the photon reversibility measurement. However, to correctly infer the multi-photon indistinguishability of the unknown state we must also assume that, as well as manipulating the external degrees of freedom reliably, the trusted interferometer acts as the identity on the internal modes. If this assumption was not made then, in the worst case, it would be possible that the characterization interferometer could act to improve the indistinguishability of the photons in the unknown state leading to an overestimate of the LOQC fidelity.

Finally, in certain circumstances, some of these trust assumptions on the interferometer can be partially relaxed. Firstly, if imperfections in the trusted interferometer occur with some random distribution that is independent of the state preparation process, then one could reasonably conclude that they would, with high probability, only act to lower the certified fidelity. Furthermore, some of the indistinguishability witnesses we analyse are either provably semi-device-independent, or can at least be shown to be so via numerical simulations of realistic noise models. In the context of our work, semi-device-independence is the property that if the characterization interferometer does introduce errors in the external degrees of freedom, these act only to lower the certified fidelity threshold. This strengthens the case for certification as a well-posed problem: using a semi-device independent protocol removes the need for assumptions on the degree of control over the interferometer.

In addition to trust assumptions for the certification interferometer, we may also ask for some additional assumptions on the unknown input state. Whilst measuring $\Pi_\mathrm{sym}^\mathbf{1}$ in an LOQC experiment is challenging in the most general case, we can obtain some useful bounds in certain instances by utilizing the substantial literature on so-called \emph{indistinguishability witnesses} \cite{Brod:2019cf,Pont:2022,vanderMeer:2021,Somhorst:2023}. Essentially, each of these witnesses takes as input a photonic state that is promised to be an $n$-photon single-occupation state, $\rho_{\mathbf{1}_n}$ and outputs a lower bound to the weight of the projection onto the symmetric subspace. These witnesses must also assume the input has only a somewhat simplified form of distinguishability noise acting on the internal degrees of freedom which we now explain. 

The effects of partial distinguishability on multi-photon interference experiments are very subtle \cite{Menssen:2017,Jones:2020,Jones:2023,Shchesnovich:2018triad,Rodari:2024}. Although we have formal expressions for all probabilities, a general understanding of the observed deviations from ideal bosonic behavior, computational complexity or even the different classes of distinguishability noise is still lacking. To deal with this, many works \cite{Brod:2019cf,Giordani:2020,Pont:2022} analyse $n$-photon single-occupation states $\rho_{\mathbf{1}_n}$ with a simplified, discrete, `incoherent' distinguishability models where the state is assumed to be of the form,
 \eqn{\rho_{\mathbf{1}_n} =  c_n\rho_{\mathbf{1}_n}^{\parallel} + \sum_ic_i \rho^\perp_i, \label{Broddecomp}}
 with $\sum_i c_i=1$ and $0\leq c_i\leq 1$ where $\rho_{\mathbf{1}_n}^{\parallel} $ is a state of $n$ indistinguishable single photons and $\rho^\perp_i$ is a state where at least two of the photons are orthogonal with any non-orthogonal photons being perfectly indistinguishable and the sum runs over all of the possible configurations of this form. Given this assumption, various LOQC protocols have been devised that directly measure the coefficient $c_n$. 
 
This concept was pursued further by Annoni et al, via an illuminating result they denote the \emph{partition} representation \cite{annoni2025}.  A partition $\underline{\Lambda}$ of a set of $\Omega = \{1,2, \dots, n\}$ is a list of subsets $\underline{\Lambda} = [\Lambda_1,\Lambda_2,\dots]$ such that 
$\bigcup_i \Lambda_i = \Omega$. For $n=3$ there are 5 possible partitions given by
 \eqn{\underline{\Lambda} =\left [ (1)(2)(3), (1,2)(3), (1,3)(2),(1)(2,3),(1,2,3) \right].\label{part3}} If we consider a set of $n$ photons with a discrete distinguishability structure such that all photons are either identical or orthogonal then we see that this naturally induces a `distinguishability' partition. Considering Eq.~(\ref{part3}) in that we can interpret each partition as denoting a state
 vector $\ket{\psi_{\Lambda_i}}$ with identical photons grouped in brackets (e.g.,  the first term means that the first two photons are identical whilst the third is orthogonal, the last term corresponds to 3 identical photons and so on). We could therefore equivalently re-write the in Eq.~(\ref{Broddecomp}) as mixture of partition states (see Fig.~(\ref{fig:LOQC}a) for an $n=3$ example),
  \eqn{\rho_{\underline{\Lambda}} \ee  \sum_{\Lambda_i\in \underline{\Lambda}}p_{\Lambda_i} \ket{\psi_{\Lambda_i}}\bra{\psi_{\Lambda_i}} \label{part_state} \\
  \ee p_{(1, \dots, n)}\ket{\psi_{(1, \dots, n)}}\bra{\psi_{(1, \dots, n)}}  + \sum_{\stackrel{\Lambda_i\in \underline{\Lambda}}{/(1, \dots, n)}}p_{\Lambda} \ket{\psi_{\Lambda_i}}\bra{\psi_{\Lambda_i}}\nn.}
 In the second line, we have pulled out the term corresponding to full $n$-photon indistinguishability. For states of this form we can see that the coefficient $p_{(1, \dots, n)}$ is the same as $c_n$ in Eq.~(\ref{Broddecomp}) and can be interpreted as a weight describing the proportion of the state that will exhibit perfect $n$-photon interference and hence a bound to the probability of all photons behaving ideally in an LOQC experiment.
  
  The partition representation generalizes this further to define a class of states that can be effectively described (for LOQC purposes) as linear combinations of partition states
  \begin{definition}[Partition representation]
     An $n$-photon single-occupation state, $\rho_{\mathbf{1}_n}$ is said to have a partition representation if it is LOQC equivalent to a linear combination of partition 
     states
    \eqn{\rho_{\mathbf{1}_n} \sim \sum_{\Lambda_i\in \underline{\Lambda}}p_{\Lambda_i} \ket{\psi_{\Lambda_i}}\bra{\psi_{\Lambda_i}},}
     with $\sum_{\Lambda_i\in \underline{\Lambda}}p_{\Lambda_i}=1$. It is said to have a positive partition representation if additionally all $p_{\Lambda_i}$ are positive. Equivalently, $\rho_{\mathbf{1}_n}$ has a partition representation iff there exists a partition state, $\rho_\Lambda$ such that 
     \eqn{\tr(P_\sigma \rho_{\mathbf{1}_n} ) = \tr(P_\sigma \rho_\Lambda) \hspace{1mm} \forall \sigma .\label{part_con}}
 \end{definition} 
 Partitions also arise naturally from permutations, in that any permutation written in so-called cycle notation is a partition, e.g., the terms in Eq.~(\ref{part3}) could also be read as the cycle notation for the trivial permutation (identity), all possible pairwise swaps, and the clockwise cycle of all three elements. Given a permutation $\sigma$, let the induced partition be defined as $\underline{\sigma}$ (note that the mapping from permutations to partitions is many-to-one \cite{part_note}). One can also define the \emph{cycle structure} of a permutation $\sigma$, which we denote $\overline{\sigma}$, which is simply the number of elements affected by each disjoint cycle in the permutation, e.g. the permutation (12)(34) has a structure of two 2-cycles. Cycle structures are typically represented by graphically via Young diagrams. This is an even more `coarse-grained' property,  e.g., the permutations $\sigma_{(1,2)(3)}$ and $\sigma_{(1)(2,3)}$ have different partitions \begin{equation}\underline{\sigma_{(1,2)(3)}} \neq \underline{\sigma_{(1)(2,3)}} 
 \end{equation}
 but the same cycle structure $\overline{\sigma_{(1,2)(3)}} = \overline{\sigma_{(1)(2,3)}}$.
 
The existence of a partition representation can be related to symmetries amongst observables $\EV{P_\sigma}$ of permutations with the same partition. A state, $\rho$, is called \emph{orbit invariant} if all permutations corresponding to the same partition have the same expectation value, i.e., if $\underline{\sigma} = \underline{\tau} \Rightarrow \EV{P_\sigma}_\rho = \EV{P_\tau}_\rho$ and all orbit invariant states are guaranteed to have a partition representation. $\rho$ is said to be \emph{cycle-structure invariant} if $\overline{\sigma} = \overline{\tau} \Rightarrow \EV{P_\sigma}_\rho = \EV{P_\tau}_\rho$ and all such states possess a positive partition representation. Although many photonic states do not admit a partition representation (prominent examples include investigations of the so-called triad phase \cite{Menssen:2017,Shchesnovich:2018triad}), there are many instances where one might justifiably assume the output states possess such a description. Two prominent examples would be a collection of independent, identical photonic sources (such that $n$-photon state produced may be written as $\rho^{\otimes n}$ and simultaneously diagonalized) or incoherent noise-models such as the `orthogonal-bad-bit' model \cite{Sparrow:2017} which is widely used to model many present-day experiments \cite{renema2018efficient,Renema:2021sampling,Marshall:2024}. Lastly, and importantly for our purposes, a positive partition representation can also be enforced by implementing an average of mode permutations to obtain
\eqn{\mathcal{P}(\rho) = \frac{1}{n!} \sum_{\sigma \in \mathcal{S}_n} P_\sigma\rho P_{\sigma^{-1}} . }
This permutation-twirled state, analogous to a Pauli twirled qubit state, can be proven to be cycle-structure invariant and hence possess a positive partition representation for any initial state \cite{annoni2025}.

\begin{figure*}[htbp]
    \begin{center}
\includegraphics[width=.85\textwidth]{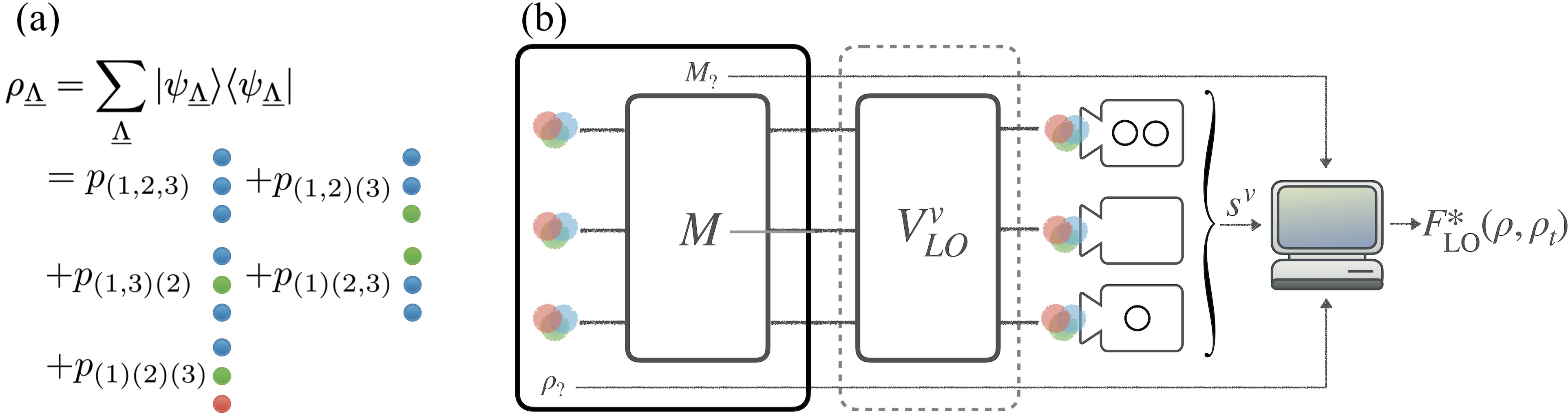}
\end{center}

    \caption{{\bf Certification setup and assumptions. (a)}  States with an incoherent, discrete distinguishability structure can be written as a linear combination of partition states which are those where all photons are partitioned into groups of identical or orthogonal internal modes, illustrated here for the case of $n=3$. Any state, $\rho$, that satisfies 
    a property called orbit invariance (see Section~\ref{sec:LOQC_fid}) is indistinguishable from some state of this form for all LOQC experiments and is said to have a positive partition representation which we denote by $\rho\sim \rho_{\underline{\Lambda}}$. We say that $\rho$ has a negative partition representation is we can find an equivalence $\rho\sim \sum_{\underline{\Lambda}} p_{\underline{\Lambda}} \ket{{\underline{\Lambda}}}\bra{{\underline{\Lambda}}}$ where $\sum_{\underline{\Lambda}} p_{\underline{\Lambda}} = 1$ but not all $p_{\underline{\Lambda}}$ are 
    non-negative. {\bf(b)} The certification setup considered in Thm.~\ref{thm_wit} for witnessing the LOQC fidelity takes an unknown state preparation device (indicated by the solid black box) and evaluates its LOQC fidelity with a target state in $\rho^U_{\mathbf{1}_{n,m}} \in \mathcal{C}_\mathrm{LO}$, thereby simultaneously verifying the quality of the photonic source and interferometric processing. The unknown state is fed into a programmable, trusted interferometer although these trust requirements can be somewhat relaxed for some witnesses indicated by the grey dashed enclosing box. The interferometer will be programmed to implement the inverse of the unitary defined by the target state or one of the $\nu = \mathrm{rev}$ of the inverse followed by one of the indistinguishability witnesses $\nu \in \{\mathrm{Four},\mathrm{cyc},\mathrm{2cor},\mathrm{sHOM} \}$. The output of each measurement setting is detected using detectors that resolve the number of photons in each external mode, but not their internal degrees of freedom as indicated by the colourless circle in each detector generating outputs $\mathbf{s}^\nu$. These outputs are then classically processed to obtain a witness for the LOQC fidelity. Depending upon the witness, obtaining a bound will require additional information in the form of assumptions on the unknown state represented by the input labelled $\rho_?$ and we indicate this assumption dependence by writing $F_\mathrm{LO}^*$ for the certified LOQC fidelity. Depending upon the witness the assumption will typically be whether the the state $P_{\mathbf{1}_{n,m}}\rho P_{\mathbf{1}_{n,m}}$ has a (positive) partition representation. In Thm.~\ref{thm_exp} we present an alternative protocol where the black box of the device is opened and the indistinguishability witness directly photonic source under the stronger assumptions on the source and interferometric map $M$.}
    \label{fig:LOQC}
\end{figure*}

Returning to the topic of genuine indistinguishability measures, it becomes apparent that a unifying (sufficient) condition that allows any of them to be valid is for the input to be a single-occupation state with a (positive) partition representation and hence a well-defined value of $c_n = p_{(1,\dots ,n)}$. Then, an appropriate linear function of certain output probabilities of a well chosen interferometric experiment can be combined to obtain a witness, $\mathcal{W}$, which outputs an estimate $c_n^*$ such that
\eqn{ \tr\bk{\rho_{\mathbf{1}_n} \mathcal{W}} = c_n^*\leq c_n .\label{ind_def}}
Interestingly, the cyclic interferometer of Ref.~\cite{Pont:2022} requires only that the input state is assumed to have a (not necessarily positive) partition representation. Furthermore, we will show that the Fourier method of Ref.~\cite{Somhorst:2023} does not require the partition assumption at all (see Section \ref{Fourier_wit} and Appendix~\ref{app:TwirlingCovariance}). Instead, it gives the value of $c_n^*$ that would be observed if a partition representation had been enforced on the state via permutation twirling. Therefore, provided the input state is assumed to have a distinguishability structure such that, when projected into single occupation subspace, it possessses a positive partition representation, an indistinguishability measurement can be combined with a photon reversibility measurement as before to obtain an experimentally feasible witness as we now show. 

 \begin{theorem}[LOQC fidelity via indistinguishability witnesses]
  \label{thm_wit}Given constants $\varepsilon,\delta_1,\xi >0$ a target state $\rho_t = \rho_{\mathbf{1}_{n,m}}^U \in \mathcal{C}_\mathrm{LO}$ and $k$ copies of unknown state $\rho^\mathrm{PP}$ with a distinguishability structure such that its (normalized) projection into the single-occupation subspace in the first $n$ modes is promised to have a positive partition representation so that we may write
 \eqn{\frac{\tr\bk{P_{\mathbf{1}_{n,m}}\rho^{PP} P_{\mathbf{1}_{n,m}}}_{m-n}}{\tr(\rho^{PP}P_{\mathbf{1}_{n,m}})} \sim  c_n \rho_{\mathbf{1}_n}^{\parallel}  + \sum
 _i c_i \rho_i^\perp .\label{ppp}}
 Let $\mathcal{T}^\nu$ be a measurement setup consisting of the following.
 \begin{enumerate}
\item Detectors that can resolve photon number in external modes but are insensitive to internal degrees of freedom.
 \item A trusted, programmable linear optical interferometer. 
  \end{enumerate}
 that is used to measure,
  \eqn{\bar{p}_1 = \tr(\tilde{\rho} P_{\mathbf{1}_{n,m}})} where $\tilde{\rho} = \hat{V}_U^\dag \rho \hat{V}_U$ is experimental state `reversed' under the target interferometer,  and 
 \eqn{\bar{c}_n^* = \tr(\tilde{\rho} \mathcal{W}^\nu),}
 where $\mathcal{W}^\nu, \nu \in \{\mathrm{Four},\mathrm{cyc},\mathrm{2cor},\mathrm{sHOM} \}$ is the observable corresponding to one of the indistinguishability witnesses based upon the Fourier transform \cite{Somhorst:2023}, cyclic interferometer \cite{Pont:2022}, two-mode correlator \cite{vanderMeer:2021} or super-posed HOM dips \cite{Brod:2019cf}, with $k/2$ samples each that lower bounds $c_n$ coefficient (see Fig.~(\ref{fig:LOQC}b)). Then $\mathcal{T}^\nu$ is an ($T,\varepsilon$)-fidelity witness for the LOQC fidelity with
\eqn{T=  \bar{p}_1 + \bar{c}_n^* - 1 - \delta_1 - \xi \label{threshPP}} and a minimum number of copies
\eqn{k \geq \frac{\mathrm{ln}(2/\varepsilon)}{2} \bk{\frac{1}{(\delta_1)^2} + \frac{\kappa_n^2}{\xi^2}},}
where $\kappa_n$ is a protocol dependent coefficient that quantifies how many copies of $\rho$ are required to obtain an additive estimate of $c_n^*$ from the measured probabilities in a given witness.

Further, for any $\rho$ the bound above obtained from the Fourier method is a valid witness for the LOQC fidelity for the `twirled' state $\mathcal{P}(\rho)$ that would be obtained if the permutation twirling channel were applied to $\rho$, i.e., for any state with a partition representation the cyclic interferometer bound is a witness for $F_\mathrm{LO}(\mathcal{P}(\rho),\rho_t)$.
 \end{theorem}
 \emph{Proof:} We begin by returning to our fidelity bound
 \eqn{F_\mathrm{LO}(\rho^\mathrm{PP},\rho_t) &\geq& \tr \bk{\tilde{\rho} P_{\mathbf{1}_{n,m}}\bk{ \Pi_\mathrm{sym}^\mathbf{1}\otimes\mathbb{I}_{(m-n)} }} \nn \\
\ee \tr \bk{P_{\mathbf{1}_{n,m}}\tilde{\rho} P_{\mathbf{1}_{n,m}}\bk{ \Pi_\mathrm{sym}^\mathbf{1}\otimes\mathbb{I}_{(m-n)} }} \label{fidproof}}
 using that $P_{\mathbf{1}_{n,m}}$ squares to the identity and commutes with $\Pi_\mathrm{sym}^\mathbf{1}$. Now, have an additional assumption on the input state. Note that the condition in Eq.~(\ref{ppp}) is not the the claim that the input state itself has a positive partition representation because this requires that the input is guaranteed to have exactly one photon in the first $n$ modes. The assumption we are actually making is that, however many photons there may be in the input state, any noise processes on internal degrees of freedom are of an incoherent nature such that when the state is projected into the appropriate photon number subspace, \emph{the projected state} possess a partition representation. Furthermore, since the trusted interferometer does not disturb the internal degrees of freedom, then if the input state is promised to satisfy Eq.~(\ref{ppp}) then so too will $\tilde{\rho}$. Defining $p_1 = \tr\bk{\tilde{\rho} P_{\mathbf{1}_{n,m}}}$ we can then write this projected state with the vacuum modes traced out as as
 \eqn{\tilde{\rho}_\mathbf{1} = \tr\bk{\frac{P_{\mathbf{1}_{n,m}}\tilde{\rho} P_{\mathbf{1}_{n,m}}}{p_1}}_{m-n} \label{rproj}}
 the input assumption Eq.~(\ref{ppp}) means we can write (using condition for LOQC equivalence given in Eq.~(\ref{part_con}))
 \eqn{\tr\bk{P_\sigma \tilde{\rho}_\mathbf{1}}  = \tr\bk{P_\sigma \bk{ \tilde{\rho}_{\mathbf{1}_n}^{\parallel}  + \sum
 _i \tilde{c}_i \rho_i^\perp }} \hspace{1mm} \forall \sigma .\label{ppp1}}
 Importantly, for partition states, the quantity $c_n$ was shown to be a lower bound to $\tr(\rho \Pi_\mathrm{sum}^\mathbf{1})$ in Ref.\ \cite{annoni2025}, so we could intuitively expect these indistinguishability measures to serve as useful tools for fidelity witnesses in this case. Explicitly, we have that,
\eqn{\tr(\Pi_\mathrm{sym}^\mathbf{1}\tilde{\rho}_\mathbf{1}) \ee \tr\bk{\sum_{\sigma} \frac{P_\sigma }{n!}\tilde{\rho}_\mathbf{1}} \nn \\
 \ee \tr\bk{\sum_{\sigma} \frac{P_\sigma }{n!} \bk{\tilde{c}_n \tilde{\rho}_{\mathbf{1}_n}^{\parallel} + \sum
 _i \tilde{c}_i \rho_i^\perp}} \nn \\
 \ee \frac{n!\tilde{c}_n}{n!} + \tr\bk{\sum_{\sigma} \frac{P_\sigma }{n!} \bk{ \sum
 _i \tilde{c}_i \rho_i^\perp}}\nn \\
 &\geq & \tilde{c}_n +  \frac{\sum_i \tilde{c}_i }{n!}\nn \\
\ee \tilde{c}_n + \frac{1-\tilde{c}_n}{n!} \nn \\
 &\geq & \tilde{c}_n \nn \\
 &\geq& \tr\bk{\tilde{\rho}_\mathbf{1} \mathcal{W}} \label{ind_wit}}
where the third line follows from the fact that $\tilde{\rho}_{\mathbf{1}_n}^{\parallel}$ is completely permutationally invariant, the fourth from the fact that for any partition state $\rho_i^\perp$ it holds that $\tr\bk{\sum_\sigma P_\sigma \rho_i^\perp}\geq 1$ since even in the worst case where all $n$ photons are orthogonal the trivial permutation will still result in a positive contribution and the final line from the defintion of an indistinguishability witness. 

Substituting these expressions into Eq.~(\ref{fidproof}) we obtain,
 \eqn{F_\mathrm{LO}(\rho^\mathrm{PP},\rho_t) &\geq& \tr \bk{P_{\mathbf{1}_{n,m}}\tilde{\rho} P_{\mathbf{1}_{n,m}}\bk{ \Pi_\mathrm{sym}^\mathbf{1}\otimes\mathbb{I}_{(m-n)} }} \nn \\
 &\geq& p_1 \tr \bk{ \tilde{\rho}_\mathbf{1}\Pi_\mathrm{sym}^\mathbf{1} } \nn \\
 &\geq & p_1 \tilde{c}_n .\label{fidPP}}
This result is intuitive, representing the joint coefficient of the part on the state of having the appropriate photon number and being fully permutationally invariant (note that this is still a lower bound to the fidelity as even distinguishable states have a non-zero projection onto the symmetric subspace). Although $\tilde{c}_n$ could be lower bounded as per Eq.~(\ref{ind_wit}), we do not have capability to simultaneously project onto the single-occupation subspace whilst measuring any of witnesses. However we can obtain an experimentally achievable lower bound by writing,
 \eqn{F_\mathrm{LO}(\rho^\mathrm{PP},\rho_t) &\geq& \tr \bk{P_{\mathbf{1}_{n,m}}\tilde{\rho} P_{\mathbf{1}_{n,m}}\bk{ \Pi_\mathrm{sym}^\mathbf{1}\otimes\mathbb{I}_{(m-n)} }} \nn \\
 &\geq& \tr \bk{ \tilde{\rho} P_{\mathbf{1}_{n,m}}\bk{ \mathcal{W}\otimes\mathbb{I}_{(m-n)} }} \nn \\
 &\geq & \tr \bk{ \tilde{\rho} P_{\mathbf{1}_{n,m}}}  + \tr \bk{ \tilde{\rho} \bk{ \mathcal{W}\otimes\mathbb{I}_{(m-n)} }} \nn \\
 \ee p_1 + c_n^* - 1 .\label{fidPP_bnd}}
 
For the finite size analysis, we need to propagate the standard analysis from the actual probabilities measured in an experiment to $c_n^*$ using Hoeffding's inequality. These will be given explicitly for each witness in the following section, but in all cases, we will find that obtaining an estimate of the relevant experimental probabilities to an additive precision $\delta_2$ leads to a bound on $c_n^*$ to an additive precision of $\xi = \kappa_n\delta_2$ such that we must obtain experimental estimates with additive error scaling as $\xi/\kappa_n$. The constant $\kappa_n$, which ranges from constant scaling to exponential in $n$, therefore crucially effects the sampling complexity of each witness. Applying Eq.~(\ref{hoeff}) to obtain estimates with additive precision $\delta_1$ and $\xi/\kappa_n$ each with probability at least (1-$\varepsilon/2$) in combination with Eq.~(\ref{fidPP}) yields the final result in Eq.~(\ref{threshPP}). 

Lastly, we show in Appendix~\ref{app:TwirlingCovariance}, the Fourier transform witness is invariant under permutation twirling. Consequently, this means that the statistics and hence the bound on $c_n$ measured directly upon the input state are identical to those that would be observed for the twirled state and hence this method provides a bound on $F_\mathrm{LO}(\mathcal{P}(\rho),\rho_t)$ for any $\rho$. Similarly, it has been shown that the cyclic interferometer implements a direct measurement of a maximal cycle permutation \cite{annoni2025}, i.e.,  $c_n$, regardless of its sign (or the sign of any other partition coefficients) and as such this witness only requires the assumption of a (not necessarily positive) partition representation on the input state. 
$\square$

The protocol(s) given above allow for rigorous certification of LOQC fidelity, requiring only the assumption of a partition-like, incoherent behavior for the internal degrees of freedom of the unknown state, simultaneously certifying the correct functioning of the photonic source and the interferometer. To achieve this, the resources required were trusted, external-mode photon-number resolving detection and a programmable interferometer. More specifically, the protocol above requires a programmable unitary with a depth at least as great as the unitary used to create the target state (since this must be undone) plus the depth required to implement whichever indistinguishability witness is to be implemented on the backwards-evolved state. Now, we present an alternative protocol which reduces these optical depth requirements and provides a tighter bound at the price of some additional assumptions on the experimental state to be certified.

The essential idea is that to make use of the fact that the experiment itself is made up of two components: a photonic source that attempts to make an appropriate input state of $n$ indistinguishable single photons resulting in a state $\rho_s$; and an interferometer that attempts to implement $U$, resulting in the output state $\rho = M(\rho_s)$ (Fig.~(\ref{fig:LOQC}b)). Typically, these components that can be separately accessed and, given some additional assumptions, these measurements can be combined into an overall witness.

\begin{theorem}[LOQC witness via source measurement]
 \label{thm_exp}Given all of the same definitions as Th.~\ref{thm_wit} but with the additional assumptions that the experimental state is of the form $\rho = M(\rho_s)$ where
\begin{enumerate}
 \item the photonic source state, $\rho_s$ possesses a partition representation
 \item the interferometer map, $M$, is a passive operation that also does not alter the internal degrees of freedom, in other words, it can be written as a mixture of interferometer channels such that,
 \eqn{M(\rho) = \sum_i m_i \hat{V}_{M_i} \rho \hat{V}_{M_i} ^\dag, \hspace{1mm} M_i\in U(m), \sum_im_i=1 \label{Uassump}}
\end{enumerate} 
a test $\mathcal{T}^\nu$ that replaces the second measurement with the quantity
\eqn{c_n^* = \tr\bk{\rho_s W^\nu}}
is ($T,\varepsilon$)-witness for the LOQC fidelity with
\eqn{T = (\bar{p}_1 - \delta_1)(\bar{c}_n^*-\xi) \label{Tsource}}
\end{theorem}
\emph{Proof}: The key point is that now considering our fidelity bound given by combining Eq.~(\ref{fidPP}) and Eq.~(\ref{ind_wit}) we have that
\eqn{F_\mathrm{LO}(\rho,\rho_t) &\geq& p_1 \tilde{c}_n}
 where is $ \tilde{\rho}_\mathbf{1}$ defined by Eq.~(\ref{rproj}) and was assumed to have a positive partition representation such that $\tilde{c}_n$ is well-defined.

The time reversed state appearing in this expression can be written as \eqn{\tilde{\rho}\ee \hat{V}_{U}^\dag M\bk{\rho_s\otimes \ket{0}\bra{0}^{\otimes (m-n)}}\hat{V}_U\nn \\
 \ee \sum_i m_i \hat{V}_{U}^\dag V_{M_i} \bk{\rho_s\otimes \ket{0}\bra{0}^{\otimes (m-n)}}V_{M_i}^\dag\hat{V}_U . \label{ev}}
Now, by Assumption 2 (specifically Eq.~(\ref{Uassump})), the operator $\hat{V}_{U}^\dag V_{M_i} $ is simply a linear optical unitary that acts only upon external degrees of freedom. Furthermore, if $\rho_s$ is assumed to have a partition representation then this means it possesses exactly one photon per mode and each term in the evolution in specified in Eq.~(\ref{ev}) simply distributes these photons amongst the $m$ external modes without changing their internals states. Consequently, when the evolved state is subjected to projection $P_{\mathbf{1}_{n,m}}$, assuming that $p_1 \neq 0$ and the projection is non-trivial, the only terms that survive are external mode permutations of $\rho_s$ over the first $n$ modes and vacuum in the remaining $m-n$. Thus, we have that
\eqn{\tilde{\rho}_\mathbf{1} \ee  \tr\bk{\frac{P_{\mathbf{1}_{n,m}}\tilde{\rho} P_{\mathbf{1}_{n,m}}}{p_1}}_{m-n}\nn \\
\ee \frac{1}{p_1}\tr\left ( P_{\mathbf{1}_{n,m}} \sum_i m_i \hat{V}_{U}^\dag V_{M_i} \bk{\rho_s\otimes \ket{0}\bra{0}^{\otimes (m-n)}}\right . \nn \\
& & \left . V_{M_i}^\dag\hat{V}_U P_{\mathbf{1}_{n,m}} \right )_{m-n}  \nn \\
\ee\frac{1}{p_1} \tr \bk{\sum_im_i \sum_{\sigma_i} \mu_{\sigma_i}  \bk{P_{\sigma_i} \rho_s P_{\sigma_i} \otimes \ket{0}\bra{0}^{\otimes (m-n)}} }_{m-n} \nn \\
\ee \sum_im_i \sum_{\sigma_i} \mu_{\sigma_i}  P_{\sigma_i} \rho_s P_{\sigma_i}  }
for some coefficients $\sum_{\sigma_i} \mu_{\sigma_i}=1$. Using again the fact that $\rho_s$ is assumed to have a positive partition representation then it necessarily holds that for any permutation $P_{\sigma_i}$,
\eqn{ P_{\sigma_i} \rho_s P_{\sigma_i}  \ee  P_{\sigma_i} \bk{c_n^s \tau_n + \sum_j c_j^s \tau_j^\perp} P_{\sigma_i} \nn \\
\ee c_n^s \tau_n +  \sum_j c_j^s P_{\sigma_i} \tau_j^\perp P_{\sigma_i} }
since by definition the state $\tau_n$ is permutationally invariant and hence
\eqn{\tilde{\rho}_\mathbf{1}  = c_n^s \tau_n + \sum_i m_i \sum_{\sigma_i}\sum_j c_j^s P_{\sigma_i}\tau_j^\perp P_{\sigma_i}}
and we can see that not only does $\tilde{\rho}_\mathbf{1}$ have a positive partition representation, it has precisely the same weight for its indistinguishability coefficient such that $\tilde{c}_n = c_s$. Therefore we can substitute a measurement of an indistinguishability witness on photonic source state and write
\eqn{F_\mathrm{LO}(\rho,\rho_t) &\geq& p_1 c_s \geq p_1 \tr{\rho_s\mathcal{W}^\mu}.}
A standard analysis using Hoeffding's inequality yields the finite-size expression in Eq.~(\ref{Tsource}) completing the proof. $\square$

To summarise, we have presented a series of LOQC fidelity witnesses of increasing practicality, which can be applied to the output of any linear optics experiment either directly or, as is frequently the case, when post-selected on $n$-fold coincidences. The final theorem provides a somewhat tighter bound using shallower trusted interferometers in situations where the untrusted interferometer can be regarded as preserving both the photon number (passivity) and the indistinguishability of any input photons, or at least only introducing distinguishability errors that are negligible relative to the photon source which is a commonly made approximation. Similarly, assuming that the photon source state has a partition representation assumes that, as well as an incoherent distinguishability structure, the experimental state has negligible higher order photon number terms. Again this is generally a reasonable approximation for quantum dots and in SPDC sources this error can be arbitrarily controlled albeit at the cost of lower count rates. Several different indistinguishability witnesses have been proposed in the literature, each of which could be combined with a fixed photon reversibility measurement to obtain the LOQC fidelity witnesses given here. In the following section we compare four prominent examples.

\section{Comparing LOQC fidelity witnesses} \label{comparison}

In this section we compare four different distinguishability witnesses -- published in Refs.\ \cite{Brod:2019cf,vanderMeer:2021,Pont:2022,Somhorst:2023} -- that get promoted to fidelity witnesses via Eq.~(\ref{fidPP_bnd}) and Eq.~(\ref{fidPP}) respectively, depending on the technical setup at hand. First, we describe each witness and revise how it estimates or lower-bounds the indistinguishable weight $c_n$. Then, we analyse the witnesses based on a set of criteria that capture both their theoretical assumptions and their practical performance in linear optical experiments.

The first criterion is \emph{tightness}, describing how close the witness can, in principle, approach the true LOQC fidelity. A tighter witness yields a stronger and more informative lower bound, reducing the gap between the certified value and the underlying fidelity to the target linear-optical computation.

The second criterion is \emph{semi–device independence} \cite{PhysRevLett.125.150503,Blind}, which, in our setting, quantifies the robustness of a witness to deviations in the implementation of the interferometer. A reliable witness should not overestimate the true LOQC fidelity even when the optical network departs from the ideal one. This form of stability is highly desirable for certification methods intended for realistic photonic platforms.

Each witness also relies on structural \emph{assumptions about the input states}. To make these explicit, we distinguish three levels of assumptions: (i) no assumptions, (ii) witness requires a state with a partition representation, and (iii) witness requires a state with a \emph{positive} partition representation.

Finally, we assess the \emph{sample complexity} of each method, i.e., how many measurements are required to estimate the witness with a prescribed error confidence $1-\epsilon$. We can think of $c_n(\epsilon)$ as a lower bound on $c_n$ with confidence $1-\epsilon$ and write
\begin{equation}
    c_n(\epsilon) \geq c_n^* - \xi(N,\epsilon)
\end{equation}
where $c_n^*$ is the measured value and $\xi(N,\epsilon)$ is the generalized statistical error we would like to control by choosing an appropriately large sample size $N$. This generalized statistical error can be calculated using Hoeffding's inequality \cite{Hoeffding1963}.
To quantify the difference between the methods, we split the generalized statistical error into the witness-dependent factor $\kappa_n$ and the statistical error in the observables, $\delta$ 
\begin{equation}
    \xi(N,\epsilon) \equiv \kappa_n \delta(N,\epsilon)
\end{equation}
and compute this witness-dependent factor. For the complete definition and calculations of the witness-dependent factors, see Appendix \ref{app:sample_complexity}. 

These criteria jointly characterize the trade-offs between generality, accuracy, and practical feasibility, providing a clear and operationally relevant basis for comparing fidelity witnesses for photonic linear-optical quantum computing. We will now describe in turn how each witness functions and its performance relative to these criteria with the result summarized in Table~\ref{tab:witness_comparison}.\\


\begin{figure*}[htbp]
    \centering
    \includegraphics[width=0.74\linewidth]{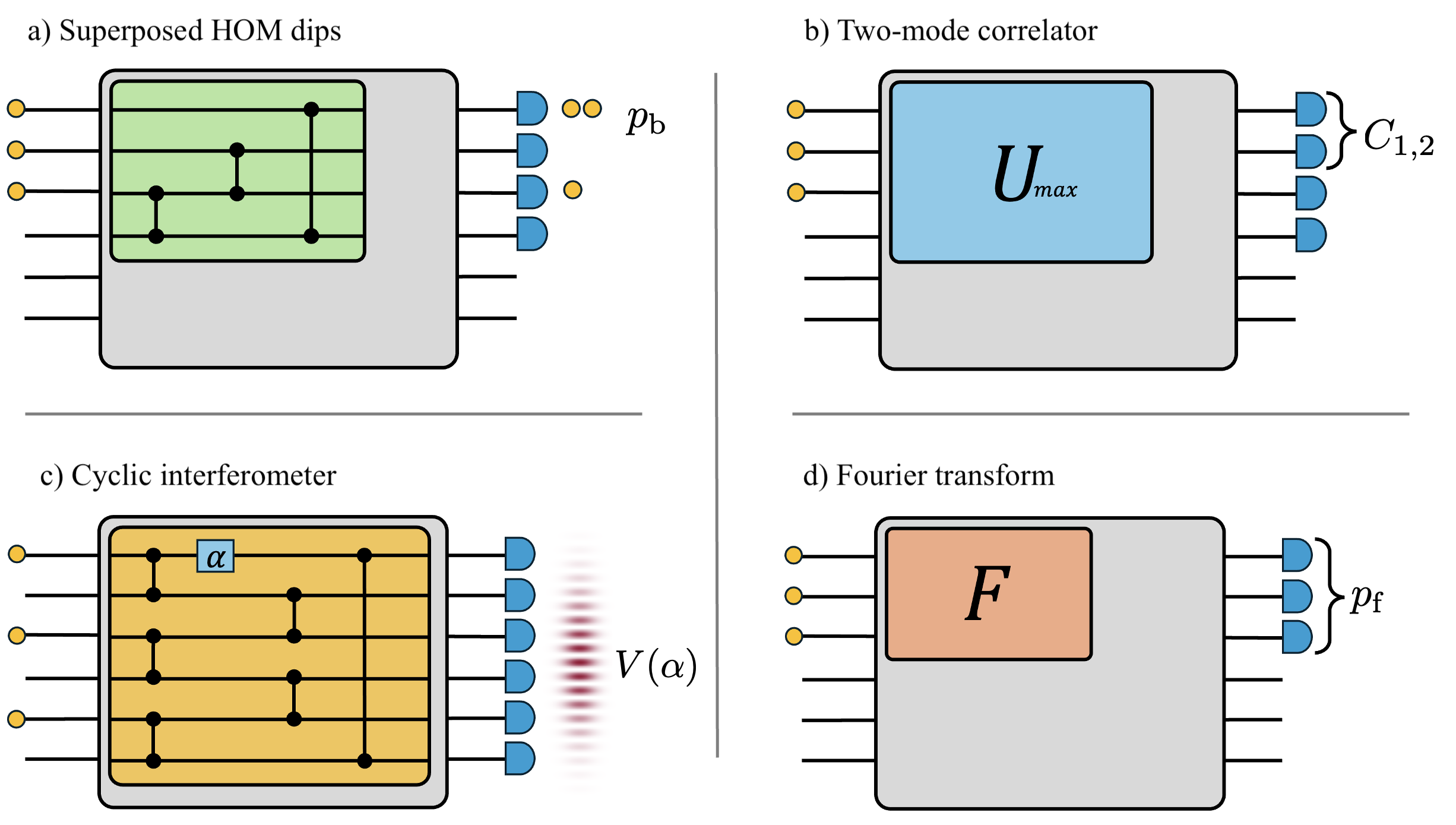}
    \caption{
    \textbf{Interferometers for the four fidelity witnesses.}
    \textbf{(a)} \emph{Superposed HOM dips:} A network of balanced beam splitters implementing multiple Hong-Ou-Mandel-type interference experiments between a reference photon, that was evenly split by a discrete Fourier transform, and the remaining photons. The witness is based on the measured bunching probability \(p_{\text{b}}\).
    \textbf{(b)} \emph{Two-mode correlator interferometer:} A programmable multiport unitary optimized to maximize two-mode number correlations \(C_{i,j}\), providing an estimate of pairwise interference contributions.
    \textbf{(c)} \emph{Cyclic interferometer:} A shallow two-layer beam-splitter network with a tunable phase \(\alpha\), arranged such that photons interfere along a cyclic pattern. The resulting interference fringe yields a visibility \(V(\alpha)\) that probes cyclic permutation symmetry.
    \textbf{(d)} \emph{Fourier interferometer:} A discrete Fourier transform acting on the external modes, exploiting exact many-particle suppression laws. The witness is obtained from the total probability \(p_\mathrm{f}\) of strictly suppressed output events.
    }
    \label{interferometers}
\end{figure*}

\subsection{Superposed HOM dips}

The witness proposed by Brod {\emph{et al.}}~\cite{Brod:2019cf} is based on a sequence of balanced beam splitters implementing a quantum Fourier transform followed by a set of pairwise \emph{Hong-Ou-Mandel} (HOM) interference experiments, as illustrated in Fig.~(\ref{interferometers}a). The measured observable is the bunching probability $p_{\text{b}}$, defined as the probability that photons emerge bunched at the output ports of the interferometer. The protocol relies on performing HOM-type interference between a chosen reference photon and each of the remaining photons, effectively probing a collection of two-photon interference visibilities.

A key structural feature of this approach is the choice of a reference photon. All interference information is extracted relative to this photon, while interference between non-reference photons is not directly accessed. As a consequence, the witness probes only a specific subset of pairwise overlaps rather than symmetrically testing all photon pairs. This asymmetry will be relevant when discussing the robustness of the witness under realistic distinguishability models.

Ref.~\cite{Brod:2019cf} showed that for any state in which at most $n-1$ photons are mutually indistinguishable, the bunching probability is upper bounded by
\begin{equation}
p_{\text{b}} \le p^{\ast} := \frac{2^{n-3}}{2^{n-2}} .
\end{equation}
Exceeding this threshold, i.e., $p_{\text{b}} > p^{\ast}$, certifies genuine $n$-photon indistinguishability. For states that can be written as a convex mixture over distinguishability partitions (see also Eq.~(\ref{Broddecomp})),
\begin{equation}
\rho_{\mathbf{1}_n} = c_n \rho_{\parallel} + (1 - c_n)\rho_{\perp},
\end{equation}
where $\rho_{\parallel}$ denotes the fully indistinguishable state and $\rho_{\perp}$ a classical mixture over all other partitions, the bunching probability takes the form
\begin{equation}
p_{\text{b}} = c_n + (1 - c_n)p^{\ast}.
\end{equation}
Solving for $c_n$ yields the bounds
\begin{equation}
\frac{p_{\text{b}} - p^{\ast}}{1 - p^{\ast}} \le c_n \le 2p_{\text{b}} - 1 ,
\end{equation}
which provides a lower bound on the weight of the fully indistinguishable partition. This bound is in general not tight (only when $c_n$ is close to one), as is shown in Fig.~\ref{fig:comparisionfigure}(e). Our numerical simulations, shown in Fig.~\ref{fig:comparisionfigure}(d), further indicate that the witness is robust to small random perturbations of the ideal interferometer, indicating semi-device independent, although we do not prove this analytically.

The validity of the witness relies crucially on the assumption that the input state admits a \emph{positive} partition representation. This convex structure ensures that interference contributions from different distinguishability sectors add incoherently, allowing the bunching probability $p_{\text{b}}$ to be interpreted as a monotonic function of $c_n$ (monotonicity is important here because it allows us to invert a single measured probability into a rigorous lower bound on $c_n$). If this assumption is violated, the inequalities above are no longer guaranteed to hold.

In experimentally relevant time-delay models, partial distinguishability is introduced through relative temporal shifts between photons, see Fig.~(\ref{fig:comparisionfigure}b). Such models naturally produce states that admit a partition representation but not necessarily a positive one. In this regime, the superposed HOM dips witness can overestimate $c_n$ and even exceed its ideal value, thereby falsely certifying genuine $n$-photon indistinguishability, see Fig.~(\ref{fig:comparisionfigure}c). This caveat is partially related to the dependence on the chosen reference photon. In a simple time-delay model for photon distinguishability, an unfortunate choice of reference photon can already lead to the failure of the witness. In Appendix \ref{app:part_weight} we provide an explicit example illustrating the behavior shown in Fig.~(\ref{fig:comparisionfigure}b/c).

While the witness is efficient in terms of sample complexity, with a generalized statistical pre-factor $\kappa_n = 2n - 2$, this efficiency comes at the cost of strong structural assumptions on the input state.

\subsection{Two-mode correlator}\label{sec:two_mode_correlator}
The two-mode correlator witness, proposed by van der Meer \emph{et al.}~\cite{vanderMeer:2021}, is based on measuring second-order number correlations between selected output modes of a programmable multiport interferometer, as illustrated in Fig.~(\ref{interferometers}d). 
The interferometer is chosen to maximize the sensitivity of a two-mode correlator
\begin{equation}
C_{i,j} = \langle n_i n_j \rangle - \langle n_i \rangle \langle n_j \rangle ,
\end{equation}
where $n_i$ denotes the photon number operator in output mode $i$. 
For perfectly indistinguishable photons propagated through a particular interferometer \(U_{\max}\), the correlator attains its maximum possible value
\begin{equation}
C_q^{(n)} = \frac{1}{4} - \frac{1}{2n},
\end{equation}
whereas for fully distinguishable photons the correlator is non-positive, \(C_d \leq 0\) for any interferometer.
Therefore, observing $C_{i,j}>0$ witnesses two-photon interference and it can also be shown that observing $C_{i,j}>C_q^{(n-1)}$ witnesses genuine $n$-photon interference, regardless of what interferometer was actually implemented, i.e., in a semi-device-independent manner, as was shown in the original work \cite{vanderMeer:2021}. Naturally, actually implementing the interferometer \(U_{\max}\) is the optimal choice to maximize the contrast between different interference regimes, enabling the correlator to act as a maximally robust indicator of multi-photon interference.

The authors of Ref.~\cite{vanderMeer:2021} only go as far as using the two-correlator itself as a statistically significant witness of multi-photon interference. To convert $C_{i,j}$ into a bound on $c_n$, we upper bound the correlator attainable by states with $c_n=0$ (i.e.,  with no fully symmetric $n$-photon component). Denoting this worst-case value by $C_{\max}^{(n-1)}$, positivity of the partition representation implies a worst-case decomposition for the observed correlator of
\begin{equation}
C_{i,j} = c_n\, C_q^{(n)} + (1-c_n)\, C_{\max}^{(n-1)} ,
\end{equation}
and therefore
\begin{equation}\label{eq:twomodemax}
c_n \ge \frac{C_{i,j}-C_{\max}^{(n-1)}}{C_q^{(n)}-C_{\max}^{(n-1)}},
\end{equation}
which is a direct lower bound on the symmetric weight and, consequently, on the LOQC fidelity. Note that $C_q(n)$ is the maximum possible value of the two-mode correlator for any state achieved by $n$ photons injected in $U_{\mathrm{max}}$ and adding a single distinguishable photon to $n-1$ identical ones can only reduce the value of $C_{ij}$ as per Eq. (2) of Ref.~\cite{vanderMeer:2021}. It follows that $C_{\max}^{(n-1)}\leq C_q^{(n-1)}$. For the special case of $n=3$ (the case we experimentally implement in this work) we get that $C_q^{(n-1)}=C_q^{(2)}=0$ and thus Eq.~(\ref{eq:twomodemax}) reduces to the simplified form

\begin{equation}
c_n \geq \frac{ C_{i,j}}{C_q^{(n)}}.
\end{equation}

As shown in Fig.~(\ref{fig:comparisionfigure}e), this bound is not tight. It is, however, analytically proven to be semi-device independent, which is also reflected in our numerical simulations in Fig.~(\ref{fig:comparisionfigure}e). It is worth mentioning that while the bound on $c_n$ requires a positive partition representation, it remains an open problem whether the original two-point correlator witness does.

The validity of this witness relies on the positivity of the partition representation. If the state contains coherences between different distinguishability sectors, or if some partition weights are negative - as occurs, for example, in realistic time-delay models mentioned earlier - then the correlator need not vary monotonically with $c_n$, and the bound may fail.
Conceptually, this limitation arises because the witness projects the full many-body interference structure onto a single pairwise observable, thereby discarding higher-order interference contributions.

From an experimental perspective, the two-mode correlator witness is efficient due to its low circuit depth and the fact that it only requires access to second-order correlations.
The generalized statistical pre-factor scales polynomially as $\kappa_n = \mathcal{O}(n^2)$, reflecting the growing variance of number correlations (see Appendix~\ref{app:sample_complexity}).

\subsection{Cyclic interferometer}
The witness proposed by Pont \emph{et al.}  \cite{Pont:2022} uses a cyclic 2-layer beam-splitter network with $2n$ modes and a tunable phase $\alpha$, as shown in Fig.~(\ref{interferometers}b). This work specifies the input-output pair of mode assignment lists $(\mathbf{g},\mathbf{h})$; The input $\mathbf{g}$ must have exactly one photon in each pair of adjacent modes with grouping $(1,2)(3,4)(2n-1,2n)$, and the output $\mathbf{h}$ must have exactly one photon in each pair of adjacent modes with grouping $(2,3)(4,5)(2n,1)$. For any such valid input–output pair $(\mathbf{g},\mathbf{h})$ the resulting interference fringe is
\begin{equation}
    P(\alpha)= 2^{-(2n-1)}[1+(-1)^{n+p+q} c_n \text{cos}\alpha],
\end{equation}
where $p$ is the number of occupied even modes in the \emph{input} and $q$ is the number of occupied even mores in the \emph{output}. For the purposes of this work, we implement the simpler case where the input has one photon in each odd-numbered mode, $\mathbf{g}=(1,3,5, \dots,  2n-1)$. In that case $p=0$ and the interference fringe takes the form
\begin{equation}\label{eq:pontwitness}
    P(\alpha)= 2^{-(2n-1)}[1+(-1)^{n+q} c_n \text{cos}\alpha].
\end{equation}

The witness uses the visibility of the interference fringe $V=(P_{\rm max}-P_{\rm min})/(P_{\rm max}+P_{\rm min})$, defined here for a single output configuration, which is exactly identified with $c_n$ via Eq.~(\ref{eq:pontwitness}). For the input state $\mathbf{g}=(1,3,5,\ldots,2n-1)$, there are exactly $2^n$ eligible output configurations. By symmetry, half of these configurations ($2^{n-1}$) have an odd number of occupied even modes and half have an even number. Each individual configuration with an odd (even) number of occupied even modes occurs with probability $(1\pm c_n\cos\alpha)/2^{2n-1}$ respectively. Summing over all configurations of a given parity therefore yields total probabilities
\begin{align}
P_{\mathrm{odd}}&=\frac{2^{n-1}(1+c_n\cos\alpha)}{2^{2n-1}}=\frac{1+c_n\cos\alpha}{2^n}, \\
P_{\mathrm{even}}&=\frac{2^{n-1}(1-c_n\cos\alpha)}{2^{2n-1}}=\frac{1-c_n\cos\alpha}{2^n}.
\end{align}

Choosing a fixed phase $\alpha$ that simultaneously maximizes $P_{\mathrm{odd}}$ and minimizes $P_{\mathrm{even}}$ allows one to identify these quantities with $P_{\max}$ and $P_{\min}$, from which the visibility directly yields $V=c_n$. Under the assumption that the input state can be written in a partition representation, the distinguishability witness thus takes the form
\begin{equation}
    c_n \geq V^*
\end{equation}
where $V^*$ is the experimentally estimated visibility $V$ for all eligible output configurations. This bound is tight, see Fig.~(\ref{fig:comparisionfigure}e), because the measured quantity (the visibility) is precisely equal to $c_n$. However, our numerical tests have shown that the witness is not semi–device independent. Random perturbations of the interferometer frequently overestimate the value of $c_n$ leading to overestimations of the true LOQC fidelity, as can be seen in Fig.~(\ref{fig:comparisionfigure}d).

When it comes to restrictions on the input state, the Cyclic Interferometer witness has a broader validity range than originally stated. Because it directly measures the coefficient associated with the full-cycle partition \cite{annoni2025}, which is directly linked to the visibility and therefore always non-negative, it is valid for any state with a partition representation, including those with negative coefficients. However, the witness also comes with a caveat. The interferometer is relatively shallow. As a consequence, the observed interference fringe is not sensitive to all pairwise distinguishabilities between photons. In particular, within the commonly used model of independent pure internal states the expectation value of the $n$-cycle permutation $(1,2,\dots,n)$ takes the form
\begin{equation}
\langle P_{(1,2,\dots,n)}\rangle
=\prod_{k=1}^{n}\langle\phi_k|\phi_{k+1}\rangle,
\qquad \phi_{n+1}\equiv\phi_1,
\end{equation}
and depends only on overlaps between neighbouring photons along the cycle. Overlaps between photons that are not adjacent in the cycle ordering do not enter this expression. This implies that a high fringe visibility does not, in general, guarantee mutual indistinguishability among all photon pairs, but only among those that are within a two-layer interferometer's distance apart from one another. For example, in the case of $n=4$ the cyclic interferometer measures the following expectation value
\begin{equation}
\bigl\langle P_{(4,1,2,3)}\bigr\rangle
=\braket{\phi_4|\phi_1}\braket{\phi_1|\phi_2}\braket{\phi_2|\phi_3}\braket{\phi_3|\phi_4}.
\end{equation}
Crucially, the overlap $\braket{\phi_4|\phi_2}$ does \emph{not} appear in this product. This means, we can engineer a set of photon states specifically such that all photons except two and four have a non-zero overlap, yet $\braket{\phi_4|\phi_2}=0$, see Fig.~(\ref{fig:comparisionfigure}a). The witness is supposed to certify genuine $n$-photon distinguishability, and hence its outcome in this case should be zero. But since the overlap $\braket{\phi_4|\phi_2}$ does not appear in the cyclic product that constitutes the witness, this pair of fully distinguishable photons will not be detected and thus the witness falsely certifies $n$-photon indistinguishability.

At first sight, Eq.~(\ref{eq:pontwitness}) appears to describe an exponentially small probability per output event, scaling as $2^{-(2n-1)}$. At the same time, there are exponentially many output configurations $\mathbf{h}$ that satisfy the post-selection condition of having exactly one photon per output mode pair. One might therefore wonder whether these two exponentials cancel, yielding an overall probability that is constant in $n$. This is not the case. The total probability of obtaining any `useful' event, i.e., any output configuration that contributes to the witness, is
\begin{equation}
P_{\mathrm{tot}} = P_{\mathrm{odd}} + P_{\mathrm{even}} = \frac{2}{2^n} = \frac{1}{2^{n-1}},
\end{equation}
which still decays exponentially with photon number. Thus, although the exponential suppression of individual output probabilities is partially compensated by an exponential number of valid output patterns, the cancellation is incomplete. The Cyclic Interferometer witness therefore remains exponentially inefficient in terms of sampling complexity, requiring $O(2^{n-1})$ experimental runs to observe a constant number of relevant events. This scaling is reflected in the generalized statistical pre-factor $\kappa_n = 2^{n-1}$.

\subsection{Fourier transform \label{Fourier_wit}}

The witness proposed by Somhorst \emph{et al.}  \cite{Somhorst:2023} implements a discrete Fourier transform $U_F$ on the external modes, see Fig.~(\ref{interferometers}c).
The central observable of the Fourier transform witness is the total probability $p_{\text{f}}$ of output configurations that are \emph{forbidden} (suppressed) for perfectly indistinguishable bosons traversing a discrete Fourier interferometer. These suppression laws, first identified in \cite{Tichy:2012gt} and previously proposed to distinguish fully distinguishable \cite{Tichy:2014} and $(n-1)$-distinguishable behavior \cite{Stanisic:2018}, arise from exact cancellations of many-body interference amplitudes enforced by the cyclic symmetries of the Fourier matrix. Importantly, suppression is an exact property for fully symmetric bosonic states and is violated as soon as permutation symmetry is broken by partial distinguishability. 
This makes $p_{\text{f}}$ simultaneously a coarse-grained yet highly sensitive and powerful diagnostic: It vanishes identically for all states in the target class $\mathcal{C}_{\mathrm{LO}}$, while increasing monotonically with the weight of non-symmetric components. As a consequence, the Fourier transform witness yields a direct and tight bound on the symmetric weight $c_n$ without requiring any reference photon, pairwise visibility model, or positivity assumption on a partition representation.

It is important to clarify the precise sense in which the Fourier transform witness applies to general input states, and how this relates to earlier work. In Ref.\ \cite{Somhorst:2023}, an $n$-photon indistinguishability witness has been defined by bounding the indistinguishable weight $c_n$ (there the protocol is is described in terms of bounding the distinguishable weight) in terms of the total probability $p_{\text{f}}$ of forbidden output event. The analysis proceeds by noting that a detector that sums over all internal degrees of freedom can be thought of as a fictitious internal-resolving detection in an arbitrary internal basis that is subsequently coarse grained over. The calculation proceeds relative to the (unknown) internal state of one of the detected photons and looks to bound the probability that a distinguishable photon would be detected in the same run. This anticipates our notions of the LOQC fidelity, already recognising that it is not the fidelity to a particular internal state but rather a relative quantity that is most relevant. The fidelity witness is derived by noting that the weight of the distinguishable subspace can be bounded by monitoring the observed probability of forbidden patterns, $p_\mathrm{f}$ and assuming all such events are caused by an element in the distinguishable subspace (which we would here call a partition) with the lowest probability of triggering a forbidden pattern, thereby attributing a worst case, maximal value of ($1-c_n$) for an observed $p_\mathrm{f}$.

In our current notation, the witness from Ref.\ \cite{Somhorst:2023} takes the form
\begin{equation}\label{eq:lambda}
c_n \ge 1 - \frac{p_{\text{f}}}{\lambda},
\end{equation}
where the pre-factor $\lambda$ was defined as a worst-case (smallest) probability of triggering forbidden output over all admissible input states. Although the bound was defined in this way for arbitrary $n$, an explicit value of $\lambda$ was derived only for the $n=3$ (note that, in fact, an arithmetical error lead to a reported value $4/9$, whereas it should have been $2/3$, such that the indistinguishabilities are actually underestimated). Here, we compute the scaling of the pre-factor $1/\lambda$, which is currently available only for specific noise models.

In particular, for a subset of partition states where at least one photon is orthogonal to all the others (this includes the \emph{orthogonal-bad-bit} (OBB) model) we derive an explicit expression for $\lambda_n$. The bound has the form (see Appendix \ref{app:OBB})
\begin{equation}
    c_n \geq 1-\frac{n}{n-1}p_{\text{f}}
\end{equation}
or in the $3$-photon case we implemented
\begin{equation}
    c_n \geq 1-\frac{3}{2}p_{\text{f}}.
\end{equation}


As shown in Fig.~(\ref{fig:comparisionfigure}d/e), for an ideal Fourier interferometer the witness provides a relatively tight lower bound on the LOQC fidelity for a broad class of input states. In addition, the witness is expected to be conservative with respect to imperfections in the implemented unitary. By analogy with the two-photon Hong-Ou-Mandel effect at a balanced beamsplitter, we conjecture that the Fourier transform maximizes the forbidden event probability $p_{\text{f}}$ for a given level of partial distinguishability. Under this assumption, deviations from the ideal Fourier interferometer can only reduce $p_{\text{f}}$, ensuring that the inferred fidelity is not overestimated. While a general proof of this extremal property is currently lacking, our numerical perturbation analysis in Fig.~(\ref{fig:comparisionfigure}d) provides supporting evidence for this behavior under small unitary deviations.

The Fourier interferometer witness takes advantage of the symmetries of the discrete Fourier transform to construct a fidelity estimator with minimal structural assumptions on the input state. Its validity does not depend on the existence of a partition representation, making it the most state-assumption-free witness considered here. The key conceptual reason for this robustness is that $p_{\text{f}}$ is invariant under \emph{permutation twirling},
\begin{equation}
p_{\text{f}}(\rho) = p_{\text{f}}(\mathcal{P}(\rho)),
\end{equation}
where $\mathcal{P}$ denotes the averaging over all external mode permutations. As shown in Appendix \ref{app:TwirlingCovariance}, this invariance implies that the Fourier transform witness depends only on orbit-averaged permutation invariants. Operationally, this means that the witness behaves \emph{as if} it were acting on a partition-representable state, even when the underlying physical state is fully general and may contain coherences between distinguishability sectors. This property is unique among the witnesses we consider.


For the states where we have derived an explicit expression for $\lambda_n$ we can directly compute the sample-complexity (see Appendix \ref{app:sample_complexity}) and find that the witness remains sample-efficient in $n$, with a generalized statistical pre-factor 
\begin{equation}
\kappa_n = \frac{n}{n-1}.
\end{equation}
For general partition states beyond the OBB model, the existence of a universal worst-case pre-factor remains an open problem. However, we conjecture that the favorable scaling also holds more generally. This would make intuitive sense because the partition with the lowest probability of causing a forbidden state would presumably be the one where $n-1$ photons are indistinguishable and only one is orthogonal, which is precisely the OBB model covered in our analysis (see also added note). From this line of reasoning, we conjecture that the OBB model provides a worst-case pre-factor and that the scaling will be at least as efficient for other noise models.\\

To summarize, in this section we have carried out a systematic theoretical and numerical comparison of four experimentally feasible indistinguishability witnesses based on two-mode correlations, superposed HOM interference, cyclic interferometry and Fourier suppression. By analyzing tightness, semi-device independence, assumptions on the input state and sample complexity, we have identified clear qualitative and quantitative trade-offs between the different approaches. The numerical results are summarized in Fig.~(\ref{fig:comparisionfigure}), which illustrates the relative performance of the witnesses, and in Table~\ref{tab:witness_comparison}, which condenses the key properties and resource requirements. This comparison shows that while several witnesses can provide meaningful bounds in restricted regimes, the Fourier-based witness achieves the most favorable overall balance between tightness, robustness, and scalability, establishing it as the most operationally meaningful fidelity witness among those studied here. This makes it a natural benchmark for LOQC certification protocols.

\begin{figure*}
    \centering
    \includegraphics[width=1\linewidth]{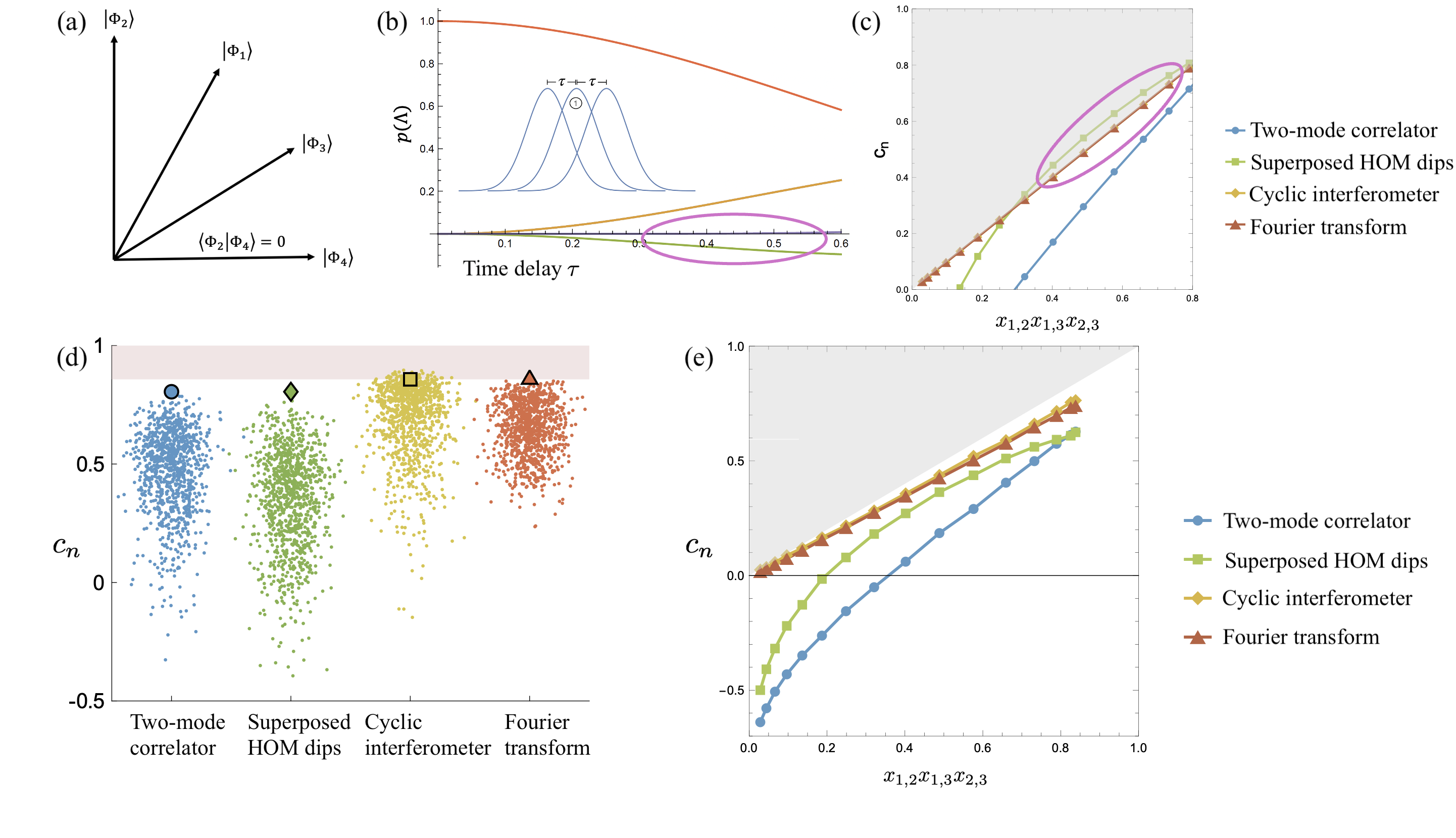}
    \caption{\textbf{Witness comparison.} \textbf{(a)} Geometric depiction of an $n=4$ photon states and their respective angles. This is an example of an n-photon state where one photon pair is completely orthogonal, yet it remains undetected by the Cyclic Interferometer witness. \textbf{(b)} Analytic dependence of the five partition weights (for $n=3$) on a distinguishability parameter (time delay $\tau$). For the time delay model discussed in the text and depicted in the center of the graph, one of the partition weights (the green one) is negative for all $\tau>0$. This can lead to the failure of some of the witnesses, which is shown in the next graph. \textbf{(c)} This graph shows the $c_n$ values in the noiseless case, $\epsilon=0$ (see (d) for noise model), for the negative partition representation state. We can see that in this case, the superposed HOM dips witness overshoots the true value of $c_n$. \textbf{(d)} Monte-Carlo analysis of device independence for the four witnesses. Each data point corresponds to a random perturbation of the ideal interferometer $U_{\mathrm{ideal}}$ of the form $U(\epsilon) = U_{\mathrm{ideal}} \cdot e^{\epsilon \log(U_{\mathrm{Haar}})},$ where $U_{\mathrm{Haar}}$ is a Haar-random unitary and $\epsilon \geq 0$ controls the perturbation strength. The big dots mark the ideal (noiseless) values for each interferometer. The shaded region is the value range above the true $c_n$ value. A semi-device-independent witness never overestimates $c_n$, corresponding to points lying below the upper boundary. Note, that the Fourier Interferometer witness is the only one that is tight and does not overshoot. The cyclic interferometer is tight but frequently overshoots, and the other witnesses do not even reach the boundary in the noiseless case (meaning that they are never tight). \textbf{(e)} This graph shows numerical simulations of the $c_n$ value as a function of the cyclic overlap (a distinguishability parameter, in this case $x_{12}x_{13}x_{23}$) for $n=3$, $m=6$, $\epsilon=0.1$, $N=500$, and HOM dip values $\textbf{x}=(0.98,0.94,0.91)$. An optimal witness would be linear in the cyclic overlap and thus follow the boundary of the gray shaded region. We can see that both the Fourier transform and the Cyclic Interferometer witness perform well in this sense.}
    \label{fig:comparisionfigure}
\end{figure*}

\begin{table*}[t]
\centering
\begin{tabular}{l|c|c|c|c|c}

Witness 
& State assumptions 
& Device independence 
& Tightness 
& Generalized pre-factor
& Computational complexity \\
\toprule
Two-mode correlator 
& Positive partition repr. 
& Yes 
& No 
& $\mathcal{O}(n^{2})$ 
& Efficient \\

Superposed HOM dips
& Positive partition repr.\ 
& Yes 
& No 
& $2n-2$ 
& Efficient \\

Cyclic interferometer
& Partition repr.\ 
& No 
& Yes 
& $2^{n}$ 
& Efficient \\

Fourier transform
& None
& Yes 
& Yes 
& $\frac{n}{n-1}$ (OBB)
& Efficient \\

\end{tabular}
\caption{
\textbf{Comparison of photonic indistinguishability witnesses.}
The table summarizes the key properties of the four witnesses considered in this work, including the structural assumptions on the input state, semi-device independence, tightness of the resulting bound on $c_n$, the generalized statistical pre-factor \(\kappa_n\) governing the sample complexity, and the computational complexity of implementing each protocol. This comparison highlights the trade-offs between generality, robustness, and experimental efficiency, and shows that the Fourier-transform-based witness uniquely combines minimal state assumptions with tight bounds and favorable scaling.
}
\label{tab:witness_comparison}
\end{table*}

\section{Experiment \label{experiment}}

In this section, we experimentally demonstrate the indistinguishability witnesses introduced above using a quantum photonic processor. The experiment is designed to allow a direct comparison of the different indistinguishability-based fidelity witnesses described in Section~\ref{comparison} under identical experimental conditions. We employ three single photons generated via spontaneous parametric down-conversion and inject them into a programmable interferometer, which enables the implementation of all four witnesses on the same hardware device. By introducing controlled temporal delays between the photons, we continuously 
tune their mutual distinguishability and probe the resulting witness values. This approach allows us to experimentally assess the tightness, robustness, and reliability of the corresponding fidelity bounds, and to test whether the theoretical advantages of the Fourier transform witness persist in an experimental setting. This section is structured as follows: first, the experimental setup is described in more detail, then the implementation of the different witnesses is described, and finally, the results are presented.



\subsection{Experimental setup}
The experiment consists of three parts which respectively generate, manipulate, and measure single photons and their associated state. These three parts are: a photon source, a photonic processor chip, and a bank of single photon detectors with readout electronics.

\subsubsection{Photon source}
Our single-photon source, visualized in Fig.~(\ref{fig:ExperimentalSetup}a), is based on a pair of \emph{periodically poled potassium titanyl phosphate} (ppKTP) crystals. These crystals are each operated in a Type-II degenerate configuration, enabling down-conversion from pump pulses of a \emph{titanium-sapphire} (Ti:Sapph) pump laser centered at $775$nm to two-mode squeezed vacuum at $1550$nm \cite{evans_2010_Phys.Rev.Lett.}, with an output bandwidth of approximately $\Delta \lambda = 20$ nm. The two arms of the produced state are then split using a polarizing beam splitter. Their purity is enhanced by passing them through bandpass filters with a bandwidth of $\Delta \lambda = 12$ nm. Next, we herald a product state of three single photons by heralding one of the arms of one of the two sources upon detecting a one-photon event, and postselecting our experiment on observing three photons at the output. We use time-of-arrival tuning to control the indistinguishability of the three-photon input state. Conditional on these
events, we prepare the input state vector $|1,1,1\rangle$ with a controllable indistinguishability \cite{tillmann_2013}. 

\subsubsection{Photonic processor}
Our photonic processor consists of an interferometer implemented using silicon nitride waveguides \cite{Triplex, QuiX2021}, with a total of $m = 12$ modes and a specified optical insertion loss (coupling plus propagation losses) of 2.5 dB (44\%) on average across the input channels. Reconfigurability is achieved through an arrangement of unit cells, each consisting of pairwise mode interactions realized as tunable Mach-Zehnder interferometers \cite{clements_optimal_2016, QuiX2021}, adjusted via the thermo-optic effect. For a complete 12-mode transformation, the average specified amplitude fidelity for randomly chosen unitaries is 
\begin{equation}
F_\text{processor} = \frac{1}{m} \mathrm{Tr} (|U^{\dagger}{\rm set}||U{\rm get}|) = 90.4 \pm 2.4\%, 
\end{equation} 
where $U_{\rm set}$ and $U_{\rm get}$ are the target and implemented transformations, respectively. However, for highly structured transformations (such as the identity), fidelities as high as 99\% are achievable. The transformation used in this work possesses significant structure—though not to the extent of the identity—and is thus expected to exhibit a fidelity lying between these two regimes. The processor also preserves the second-order coherence of the photons \cite{QuiX2021}.

\subsubsection{Detection}
Photon detection is performed using a bank of \emph{superconducting nanowire single-photon detectors} (SNSPDs) \cite{ReviewSNSPD, Marsili-SNSPD}, with standard correlation electronics for readout. Pseudo-photon-number-resolving detectors \cite{feito_measuring_2009} are assigned to each of the used output modes, with an additional click detector used to herald successful photon generation events. We postselect on triple-coincidence detection events and correct the resulting count rates to account for relative detection efficiency differences across the detector set. An additional noise source arises from higher-order pair generation combined with optical losses, which can lead to spurious postselected events. The pump power of 25 mW is chosen to balance photon generation rate and noise suppression. 

\begin{figure*}[htbp]
    \centering
    \includegraphics[width=.9\linewidth]{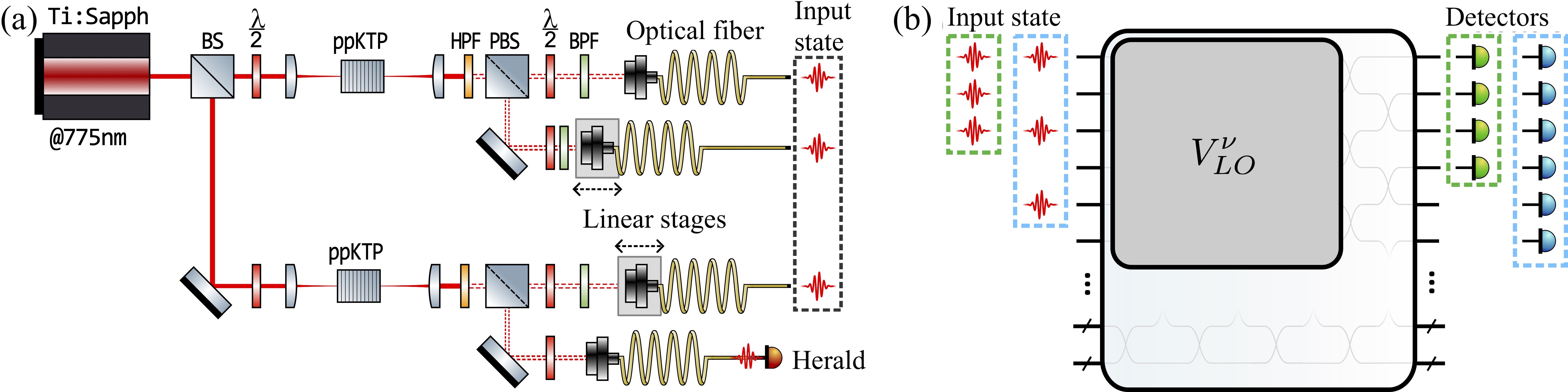}
    \caption{{\bf Experimental setup. (a)} A Ti:Sapph pump laser at 775\,nm is split by a beamsplitter into the two arms of the single photon source, and then the polarization is controlled with a half-wave plate ($\frac{\lambda}{2}$) to match the ppKTP crystal. The pumped ppKTP crystal generates photon pairs by type-II spontaneous down conversion. The pump is removed by a \emph{high-pass filter} (HPF). The photon pairs are split by polarizing beam splitters. The polarization is aligned with the fiber by another half-wave plate. Then, a \emph{band pass filter} (BPF) is used to improve the spectral purity of the photons of three of the photons that are used as the input state. The arrival time of the single photons in the input state is synchronized with the linear stages. The other photon is used to herald the single photon state. \textbf{(b)} The three single photons are used as the input state. The input is $\ket{1,1,1,0}$ (green box) for the superposed HOM dips, Fourier transform, and two-mode correlator witness. For the cyclic interferometer, the $\ket{1,0,1,0,1,0}$ (blue box) input state vector is used. The Clements decomposition is used to determine the voltages on the heater of the photonic processor to implement the necessary unitary transformation ($V_{LO}^{\nu}$) for the four different interferometers. The output is measured with a bank of superconducting nanowire single-photon detectors. For the superposed HOM, Fourier transform, and two-mode sorrelator, four three-photon pseudo-photon-number-resolving detectors (green detectors) are used, and for the cyclic interferometer, six two-photon pseudo-photon-number-resolving detectors (blue detectors) are used.}
    \label{fig:ExperimentalSetup}
\end{figure*}

\subsection{Experimental implementation}
To compare the witnesses, the photons are characterized, their overlaps are adjusted, and then the witnesses are measured. First, the characterization is described, then the implementation of the indistinguishability witnesses, and finally the demonstration of the full LOQC fidelity witness.

\subsubsection{Overlap characterization and tuning}
 On-chip measurements using the \emph{Hong-Ou-Mandel} (HOM) effect \cite{hong_measurement_1987} provide a set of calibration measurements to infer the state vector overlap between photons $x_{i,j} = \braket{\phi_i | \phi_j}$, where $\ket{\phi_i}$ represents the state vector of photon in input spatial-mode $i$. The maximum state vector overlap between photons, related to the HOM dip visibility via $V = x^2$, was measured.  After filtering, we measure visibilities of 98\%, 95\%, and 90\% for photons pairs 1\&2, 1\&3 and 2\&3, respectively (number refers to input spatial mode). It is important to note that HOM tests only provide access to $|x_{i,j}|^2$. Therefore, we make the additional assumption that $x_{i,j}$ is real. While this assumption may not hold in general, it suffices in cases where partial distinguishability is introduced by time-delays, as verified in Refs.\  \cite{Menssen:2017,rodari2024semideviceindependentcharacterizationmultiphoton}.
 To tune the indistinguishability, linear stages with micrometer precision are used to induce distinguishability in the arrival time of the single photons. The overlaps are characterized before and after each measurement of the witness to have a clear reference. 
 
\subsubsection{Comparison of indistinguishability witnesses}
The indistinguishability witnesses described in Section~\ref{comparison} are implemented in the experimental setup described above. The superposed HOM dips, two-mode correlator, and Fourier transform witness use the $\ket{1,1,1,0}$ input state vector and four three-photon pseudo photon-number-resolving detectors to discern the output click patterns of interest. The cyclic interferometer witness requires the $\ket{1,0,1,0,1,0}$ input state vector, and it uses six two-photon pseudo photon-number-resolving detectors. The four different required unitary transformations are applied to the input state by the photonic processor. From the probability distribution of the click patterns a bound for $c_n$ is found for each witness.

\subsubsection{Demonstration of the fidelity witness}
To demonstrate that the Fourier method can be used as a fidelity witnesses, we lower bound the fidelity for twelve Haar random matrices while tuning the pairwise overlap of the internal wave functions of the input photons. 
The same measurement setup is used as for the Fourier indistinguishability witness, but a different setting of the photonic processor is used. The Quix photonic processor is virtually split into two sections that both independently apply a $4\times 4$ unitary transformation. The section on the input side of the processor applies the $4\times 4$ Haar random unitary transformation to generate a state on which our method can be demonstrated. The section on the output side of the processor implements the fidelity protocol in two steps: first, the processor implements the inverse of that Haar random unitary transformation to measure the photon reversibility, and then the matrix product of the inverse Haar random unitary and the Quantum Fourier Transform is implemented to lower bound the indistinguishability.

\subsection{Experimental results}
The witnesses are measured with a 80\,MHz pump laser repetition rate. The click patterns are post-selected on three-photon events. The frequency of the three-photon click patterns is about 2\,Hz due to the probabilistic generation and photon loss. Each pairwise overlap setting of the Fourier, superposed HOM, and two-mode correlator is measured for 75 minutes, interleaving the different measurement settings to compensate for drift in the setup. The cyclic interferometer is measured for 300 minutes to acquire sufficient statistics given its poor sample complexity. The characterization of the overlaps before and after each measurement is about 30 minutes long. The Haar random unitary transformations for the demonstration of the full LOQC fidelity witness, as described in Theorem \ref{thm_exp}, are each measured for 20 minutes with the same characterization time as the indistinguishability witness.

Fig.~(\ref{fig: experimental results}a). shows the comparison of the indistinguishability witnesses. The markers denote the measured $c'_n$ values. The horizontal lines below each marker show the witnessed lower bound with the allowed error probability of 25\%, as calculated via Hoeffding's bound. The red-shaded area marks an overestimation of the indistinguishability. Indistinguishability below 0 is shaded gray to emphasize that this is a trivial lower bound. The experimental data is in accordance with the simulations. The Fourier and cyclic interferometer witnesses tightly bound the indistinguishability, while the superposition HOM and the two-mode correlator provide looser bounds.


In Fig.~(\ref{fig: experimental results}b) the fidelity witness that uses the Fourier protocol is demonstrated. A deduction is made to all data points such that they have 90\% certainty according to Hoeffding's inequality. The bigger markers denote the mean of the twelve matrices, and the smaller dots depict the values that were witnessed for the individual matrices. The graph part signifies that a trivial fidelity or indistinguishability is witnessed, as these values cannot be negative. The blue downwards-facing triangles show the reversibility ($p_1$); this part of the witness should be independent of the pairwise overlaps. The red upwards-facing triangles depict the Fourier indistinguishability witness ($c_n$). This witnessed indistinguishability is similar to the previous measurement of the Fourier transform witness. Then, the gate fidelity and the indistinguishability are combined per Eq.~(\ref{Tsource}) to get the LOQC fidelity as defined in 
Definition \ref{BigDefinition}. This fidelity is depicted by the magenta squares.

\begin{figure*}[htbp]
    \centering
    \includegraphics[width=0.82\linewidth]{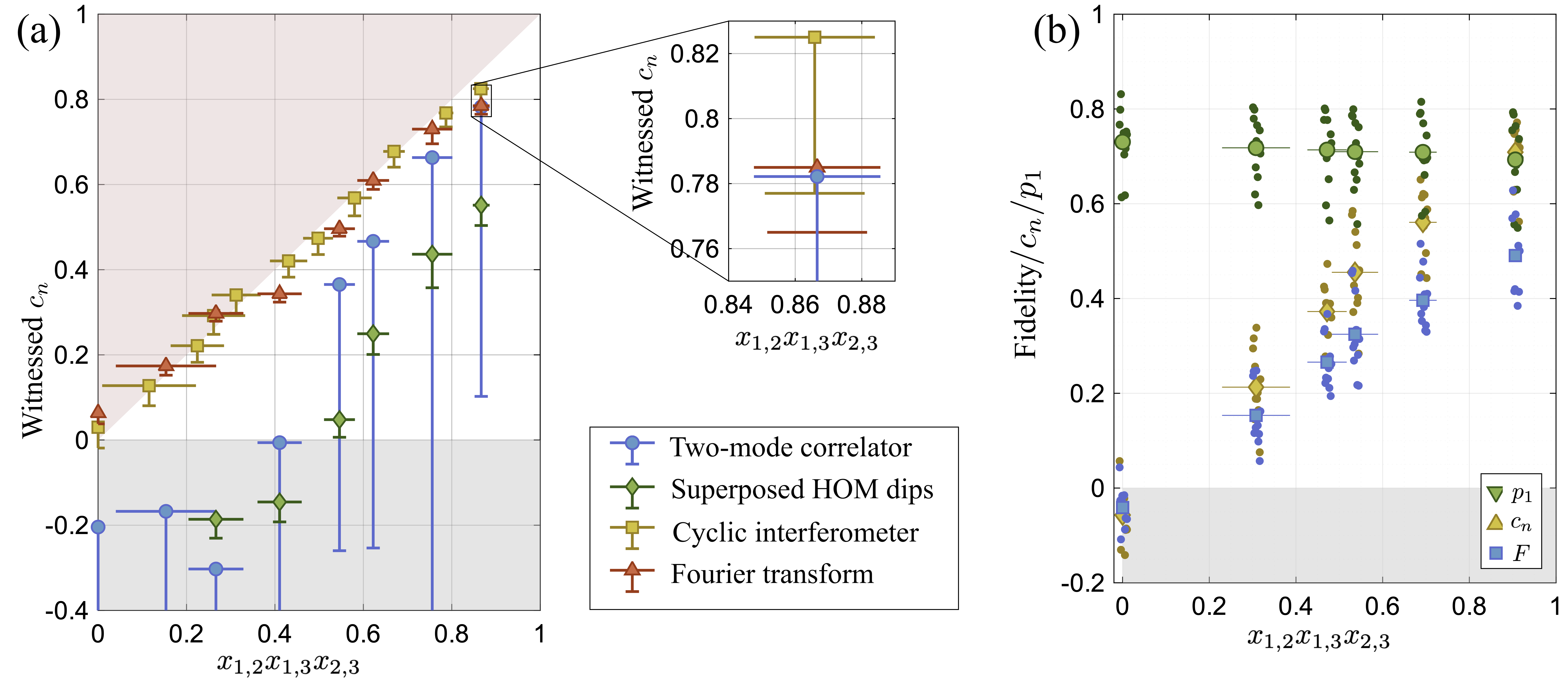}
    \caption{ {\bf Experimental results. (a)} Graph of the witnessed indistinguishability as a function of the product of pairwise overlaps of the internal states  ($x_{1,2}x_{2,3}x_{1,3}$). The markers denote the measured values for $c_n'$. The horizontal lines through these markers depict the standard deviation in $x_{1,2}x_{2,3}x_{1,3}$ and the horizontal lines under the markers denote $c_n(\varepsilon=0.25)$. The area that is marked gray does not witness indistinguishability, and the area marked by red has a higher indistinguishability than is allowed according to the theory. Overshooting corresponds to falsely overestimating the indistinguishability and therefore leads to invalid fidelity bounds.  {\bf(b)} Graph depicting the gate fidelity, indistinguishability, and fidelity witnesses of twelve Haar random matrices with $\epsilon=0.1$. The bigger markers denote the mean of the matrices, and the smaller dots depict the individual matrices. Small arbitrary offsets are added to the individual measurements for better visualization. The horizontal lines denote the standard deviation in the product of the pairwise overlaps. The gray-shaded part of the graph part signifies that a trivial fidelity or indistinguishability is witnessed.}
    \label{fig: experimental results}
\end{figure*}

\section{Discussion, conclusion, and future work}
\label{sec:discussion}

In this work we have developed a unified {operational} framework for certifying the fidelity of linear-optical quantum state-preparation protocols in the presence of partial photon distinguishability. We have clarified how traditional fidelity witnessing is neither practically achievable in standard experiments nor able to precisely capture the true resourcefulness of an experimentally prepared state. To address this, we defined an operationally well-motivated notion of fidelity that does meaningfully capture the quality of a prepared state for linear-optical quantum information processing, and then proposed several feasible methods to witness this quantity using well-calibrated interferometric measurements. Our central observation was that LOQC fidelity certification can be reduced to bounding two physically meaningful quantities: the probability that the experimentally prepared state occupies the correct external-mode single-photon subspace, dubbed photon reversibility, and a bound on the weight of the fully indistinguishable sector of the input state. 

We then investigated four methods from the literature for quantifying multi-photon indistinguishability,  experimentally implementing each witnesses using states from down-conversion photonic sources injected into a programmable, integrated interferometric chip. We complemented this with extensive theoretical and numerical analysis and evaluated their utility and limitations for fidelity witnessing, obtaining along the way some results on partial photon distinguishability of independent interest. The witnesses typically required the assumption that a positive partition representation could be associated with the prepared states - an accurate approximation for many noisy systems, and also enforceable by the application of random permutation mode permutations (twirling). We provide explicit examples of how the cyclic interferometer \cite{Pont:2022} can falsely certify fidelities when this assumption is broken. We also identify similar behavior for the superposed HOM witness \cite{Brod:2019cf} arising from a relatively simple noise model of asymmetric time delays, which is of additional interest to the study of anomalous bunching \cite{Seron:2023,Pioge:2023,Rodari:2024,Geller:2025}. On the other hand, for the Fourier method \cite{Somhorst:2023}, we show that the forbidden state probability on which it is based is invariant under permutation twirling and as such can extract useful results for any input state. For the two-mode correlation method \cite{vanderMeer:2021} the current derivation assumes a positive partition representation but we conjectured that this is unnecessary.

Regarding robustness to imperfections, the two-mode correlation method \cite{vanderMeer:2021} was analytically proven to be semi-device independent in the sense of always providing a valid lower bound even if the supposedly trusted interferometer deviates arbitrarily. We provide numerical evidence that this property also holds for the superposed HOM and Fourier witnesses and fails for the cyclic interferometer. Lastly, we addressed the critical issues of tightness and sample-complexity. Only the Fourier and cyclic methods provide in-principle exact characterizations of multi-photon indistinguishability, with the remaining methods only providing good approximations in the limit of highly indistinguishable photons. However, the cyclic method has an exponential sample complexity, rendering it unscalable to large system sizes. Conversely, we show for an error model slightly more restrictive than incoherent partition-like noise, that the Fourier method achieves $\mathcal{O}(1)$ sample-complexity and conjecture that this in fact holds generally (see note added). 

Lastly, based upon all of these considerations, we concluded that the Fourier witness represented the optimal combination of features and applied it to certify the creation of a series of Haar random states in a three-photon, four-mode system with a controlled level of distinguishability introduced via time-delays. 

Several open problems and future directions are suggested by this work. Immediate applications of these witnesses could include verified quantum sampling experiments or GHZ state creation as a resource for LOQC \cite{Bartolucci:2023,PsiQ:2025}. In this regard, it would also be of great interest to understand if they could be generalized beyond target states of definite photon number, for the purposes of certifying Gaussian \textsc{ BosonSamplers} \cite{Lund:2014,Hamilton:2017} or states for continuous variable quantum information, such as cat-states \cite{Ralph:2003} or Gottesmann-Kitaev-Preskill states \cite{GKP}. As well as the unresolved conjectures regarding the Fourier sample complexity and partial device-independence or the two-mode correlator's range of validity, other open questions include whether shadow tomography, compressed sensing, or matrix product state methods could be incorporated to render certification even more efficient in certain circumstances, and now one might extend this analysis beyond the partition-representation restrictions. 

In summary, we have provided a systematic approach to LOQC fidelity certification that unifies a variety of indistinguishability witnesses within a common operational framework. We investigated multiple specific witnesses and successfully applied them to experimentally certify complex photonic state creation. We expect that these results will be useful both for the analysis of current photonic experiments and for the development of scalable verification tools for future linear-optical quantum technologies.

\emph{Note added:} Whilst finalising this manuscript we became aware of some recent, independent work including some of the results presented here. Ref.\ \cite{Pioge:2026} conducted an extensive investigation of anomalous bunching and independently showed examples of this arising from simple time delays similar to our discussion of the super-posed HOM witness. Ref.\ \cite{Novo:2026} also considered Fourier interferometry as a version of the qubit cycle-test with applications including multi-photon benchmarking. Finally Ref.\ \cite{Sanz:2026} also experimentally implemented the Fourier based method for measuring genuine indistinguishability and showed that it is sample-complexity efficient for all partition states, thereby settling in the affirmative the conjecture that our OBB results would hold in the more general case.

\section*{Acknowledgements}

R.~S.\ is supported by the Studienstiftung des deutschen Volkes e.\ V.\  as well as by the BMFTR (PhoQuant, QPIC-1). N.~S. is supported by the Eindhoven Hendrik Casimir Institute. S.~N.v.d.~H. is supported by a contribution
from the project of the National Growth Fund program Photon Delta. S.~M.~is supported by the Connecting Industries project `Building Einstein’s Dice'. J.~E. is supported by the BMFTR (PhoQuant, QPIC-1, PasQuops, QSolid, QuSol), Berlin Quantum, the Munich Quantum Valley, the Clusters of Excellence ML4Q and MATH+, the DFG (CRC 183 and SPP 2541), and the European Research Council (ERC AdG DebuQC). J.~J.~R.~is supported by the project
``At the Quantum Edge'' 
of the research program VIDI,
which is financed by the Dutch Research Council (NWO). N.~W.~has been supported by the BMFTR (PhoQuant, QPIC-1). We thank H.M Chrzanowski for useful discussions and  J.~Saied and Frank Somhorst for their insights on the output probabilities and sample-complexity of Fourier interferometers in the orthogonal bad-bit model.

\appendix
\section{Partial distinguishability calculations \label{partial_app}}

For completeness, in this appendix we show mathematically how the photon permutation ($J_\sigma$) and mode permutation ($P_\sigma$) operators are equivalent and also derive Eq.~(\ref{Shchesnovich1}) as a special case of the results of Ref.~\cite{Shchesnovich:2015}. We begin with the observation in Appendix A of \cite{Shchesnovich:2015} that for any two permutations $\sigma$ and $\tau$ of elements indexing some list $A$
\eqn{\prod_\alpha A_{\alpha,\tau\sigma(\alpha)} = \prod_\alpha A_{\sigma^{-1}(\alpha),\tau(\alpha)}.\label{permid}}
follows from the fact that a permutation is simply a bijection between two sets of indices and that a product of scalars is independent of the order of multiplication. Note that this identity is already enough to see that the photon permutation operator defined in Eq.~(\ref{Jsig}) can be rewritten
as
\eqn{J_\sigma  \prod_{j=1}^n  \adag_{\mathbf{d_n}(j),\kappa_j} \ee  \prod_{j=1}^n \adag_{\mathbf{d_n}(j),\kappa_{\sigma(j)}} \nn \\
\ee \prod_{j=1}^n \adag_{\mathbf{d_n}(\sigma^{-1}(j)),\kappa_{j}}}
Now for states with at most one photon per mode (i.e.,  where $\mathbf{d_n}$ has not repeated entries then the permutation of the external degree of freedom of any single photon can be realized by in external mode permutation operator defined in Eq.~(\ref{Psig}) such that
\eqn{ \prod_{j=1}^n \adag_{\mathbf{d_n}(\sigma^{-1}(j)),\kappa_{j}} \ee P_{\sigma^{-1}}  \prod_{j=1}^n  \adag_{\mathbf{d_n}(j),\kappa_j} }
so we see concretely that for single-occupation states there is a mode permutation operator $P$ for any photon permutation operator $J$ as claimed in the main text.

Turning to the derivation of Eq.~(\ref{Shchesnovich1}) we start with Eqs.\ (9) and (16) from \cite{Shchesnovich:2015} re-written in our notation (in 
Ref.\ \cite{Shchesnovich:2015}, $P_\sigma$ is used to indicate photon permutation instead of $J_\sigma$) as
\eqn{P(\mathbf{s}|\mathbf{n}) \ee \frac{1}{\mu(\mathbf{n})\mu(\mathbf{s})} \sum_{\sigma_1\sigma_2} \mathcal{J}(\sigma_1,\sigma_2) \nn \\
&\times&\prod_{j = 1}^n U^*_{\mathbf{d_n}(\sigma_1(j)),\mathbf{d_s}(j)} U_{\mathbf{d_n}(\sigma_2(j)),\mathbf{d_s}(j)}\nn}
where
\eqn{\mathcal{J}(\sigma_1,\sigma_2) = \tr \bk{\hat{\Gamma}_{l_1} \otimes \hat{\Gamma}_{l_2}\otimes\cdots \otimes\hat{\Gamma}_{l_n}  J_{\sigma_2}^\dag\rho J_{\sigma_1}} }
and all sums should be understood as running over all permutations $\mathcal{S_n}$.
The $\hat{\Gamma}$ quantities represent the spectral response of the photodetectors which in the original analysis are allowed to be asymmetric. Here, we will assume all detectors have a flat response so these can all be replaced by identity operators so that we have 
\begin{equation}\mathcal{J}(\sigma_1,\sigma_2) = \EV{J_{\sigma_1\sigma_2^{-1}}} 
\end{equation}
using $J^\dag_\sigma = J_{\sigma^{-1}}$. Thus, setting $\sigma = \sigma_1\sigma_2^{-1}$ and $\tau = \sigma_2^{-1}$ and we can write
\eqn{P(\mathbf{s}|\mathbf{n}) 
\ee \frac{1}{\mu(\mathbf{n})\mu(\mathbf{s})} \sum_{\sigma} \EV{J_\sigma} \nn \\
&\times& \sum_{\tau}\prod_{j = 1}^n U_{\mathbf{d_n}(\tau^{-1}(j)),\mathbf{d_s}(j)} \prod_{j = 1}^n U^*_{\mathbf{d_n}(\sigma\tau^{-1}(j)),\mathbf{d_s}(j)} \nn \\
\ee \frac{1}{\mu(\mathbf{n})\mu(\mathbf{s})} \sum_{\sigma} \EV{J_\sigma} \nn \\
&\times& \sum_{\tau}\prod_{j = 1}^n U_{\mathbf{d_n}(j),\mathbf{d_s}(\tau(j))} \prod_{j = 1}^n U^*_{\mathbf{d_n}(j),\mathbf{d_s}(\tau\sigma^{-1}(j))} \nn \\
\ee \frac{1}{\mu(\mathbf{n})\mu(\mathbf{s})} \sum_{\sigma} \EV{J_\sigma} \nn \\
&\times& \sum_{\tau}\prod_{j = 1}^n U_{\mathbf{d_n}(j),\mathbf{d_s}(\tau(j))}  U^*_{\mathbf{d_n}(\sigma(j)),\mathbf{d_s}(\tau(j))} \nn \\
\ee \frac{1}{\mu(\mathbf{n})\mu(\mathbf{s})} \sum_{\sigma} \EV{J_\sigma} \mathrm{perm} \bk{M  *M^*_{\sigma,\mathbb{I}}},}
where we have used Eq.~(\ref{permid}) in the second and third line along with the fact that $(\sigma\tau)^{-1} = \tau^{-1}\sigma^{-1}$ and in the last line we have used the definition of the permanent
\eqn{\mathrm{perm}(M) = \sum_{\sigma \in \mathcal{S}_n} \prod_{i=1}^mM_{i,\sigma(i)}}
and our definition the matrix $M$ from Eq.~(\ref{M}) given by
\eqn{M_{i,j} = U_{\mathbf{d_n}(i),\mathbf{d_s}(j)}.}

\section{Alternative definitions of LOQC metrics \label{alt_defs}}

Here we briefly show the relationship between the LOQC fidelity and trace distance studied here and alternative quantity discussed in Eq.~(\ref{alt_cert}).

\begin{theorem}[Relationship between LOQC metrics]
The two kinds of LOQC metric satisfy the relationship
\eqn{\tilde{F}_{\mathrm{LO}}(\rho,\sigma) &\geq& F_{\mathrm{LO}}(\rho,\sigma) \nn ,\\
D_{\mathrm{LO}}(\rho,\sigma) &\geq& \tilde{D}_{\mathrm{LO}}(\rho,\sigma).}
\end{theorem}
\emph{Proof}:
For any two quantum states $\rho,\sigma$ the trace distances satisfies
\eqn{D(\rho,\sigma) \ee \max_{\{ \mathcal{E}_i\}} \frac{1}{2} \sum_i \left | \mathrm{tr}\left (\mathcal{E}_i \rho\right ) - \mathrm{tr}\left (\mathcal{E}_i \sigma\right )\right| \nn \\
&\geq& \max_{\{ F_\mathbf{s}\}} \frac{1}{2} \sum_\mathbf{s} \left | \mathrm{tr}\left (F_\mathbf{s} \rho\right ) - \mathrm{tr}\left (F_\mathbf{s} \sigma\right )\right|}
where the inequality holds since $\{\mathcal{E}_i\}$ is an arbitrary POVM which strictly includes the LOQC subset $\{F_\mathbf{s}\}$. Therefore it holds that
\eqn{D_\mathrm{LO}(\rho,\rho_t) \ee\max_{\tau \sim \rho_t} D(\rho, \tau) \nn \\ &\geq&  \max_{\tau \sim \rho_t} \tilde{D}_\mathrm{LO}(\rho, \tau) = \tilde{D}_\mathrm{LO}(\rho, \rho_t)}
where the last equality holds since, by definition, $\tilde{D}_\mathrm{LO}(\rho, \tau) = \tilde{D}_\mathrm{LO}(\rho, \rho_t)$ whenever $\tau\sim\rho_t$. An entirely analogous argument leads to the corresponding result for the fidelities.

\section{Relation between generalized indistinguishabilities and partition weights \label{app:part_weight}}

For an $n$-photon state, the \emph{generalized indistinguishabilities} are defined as expectation values of permutation operators $\hat{P}_{\sigma}$ acting on the modes
\begin{equation}
    M_\sigma = \braket{\hat{P}_\sigma}_\rho =\text{Tr}(\rho\hat{P}_\sigma), \enspace \sigma \in S_n.
\end{equation}
When a state has a \emph{partition representation} the generalized indistinguishabilities $M_\sigma$ depend only on the cycle structure of $\sigma$. If two permutations $\sigma$ and $\tau$ have the same cycle structure, denoted as $(\sigma)$ and $(\tau)$ respectively, then $M_\sigma=M_\tau$ and we can consistently denote the GI values as $M_{(\sigma)}$.\\

A \emph{distinguishability partition} $(\Lambda)$ is a set partition of $\{1, \dots,  n\}$ into blocks $\Lambda_1, \Lambda_2,\dots $, where each block corresponds to photons that are mutually indistinguishable in their internal degrees of freedom, while photons in different blocks are in orthogonal internal states. There is natural partial order on partitions by how `coarse' they are. Concretely, writing $(\Lambda)\succeq (\Pi)$ means that $(\Lambda)$ is coarser than $(\Pi)$, i.e., every block of $(\Pi)$ is contained in some block of $(\Lambda)$. Intuitively, $(\Lambda)$ allows at least as much indistinguishability as $(\Pi)$.

For a partition state vector $\ket{\psi(\Lambda)}$, 
Ref.\ \cite{annoni2025} show that the GI for a permutation $\sigma$ is simply a 0 or 1 condition,
\begin{equation}
    M_\sigma(\psi(\Lambda))=
\begin{cases}
1, \enspace \text{if } (\Lambda) \succeq (\sigma),\\
0, \enspace \text{otherwise}.
\end{cases}
\end{equation}
This means that $P_\sigma$ has expectation values $1$ iff $\sigma$ never moves an element between the distinguishability blocks $\Lambda_i$. Equivalently, each cycle of $\sigma$ must lie entirely inside some block of $(\Lambda)$. 

A state $\rho$ has a (possibly quasi-) partition representation if it is LOQC equivalent to a mixture of partition states
\begin{equation}
    \rho \sim \sum_{(\Lambda)} p(\Lambda) \ket{\psi(\Lambda)}\bra{\psi(\Lambda)}.
\end{equation}
By linearity of the expectation values, the indistinguishabilities of the mixture are
\begin{equation}
    M_\sigma = \sum_{(\Lambda)} p(\Lambda) M_\sigma(\psi(\Lambda)) = \sum_{(\Lambda)\succeq (\sigma)}p(\Lambda).
\end{equation}
In other words, if we fix a cycle-partition $(\sigma)$ then $M_{(\sigma)}$ equals to the total weight of all partition configurations $(\Lambda)$ that are coarser than $(\sigma)$. We can define the matrix
\begin{equation}
    R_{(\sigma),(\Lambda)}=
\begin{cases}
1, \enspace \text{if } (\sigma) \succeq (\Lambda),\\
0, \enspace \text{otherwise,}
\end{cases}
\end{equation}
and write the indistinguishabilities in the following matrix 
form as a function of the partition weights
\begin{equation}
    M_\sigma = \sum_{(\Lambda)} R_{(\sigma),(\Lambda)}p(\Lambda) .
\end{equation}
Since $R$ is invertible, the linear system above can be inverted to recover $p(\Lambda)$ from the measured $M_{(\sigma)}$. The inverse $R^{-1}$ is encoded by the \emph{Möbius function} $\mu$ of the partition, yielding the following inversion formula (see Appendix B of Ref.\ \cite{annoni2025})
\begin{equation}
    p(\Lambda) = \sum_{(\sigma)} \mu_{(\Lambda),(\sigma)}M(\sigma) .
\end{equation}
In the case of $n=3$ photons we have five partitions $(1,2,3), (1,2)(3), (1,3)(2),(2,3)(1),(1)(2)(3)$ and get
\[
R=
\begin{pmatrix}
1 & 0 & 0 & 1 & 0 \\
0 & 1 & 0 & 1 & 0 \\
0 & 0 & 1 & 1 & 0 \\
0 & 0 & 0 & 1 & 0 \\
1 & 1 & 1 & 1 & 1
\end{pmatrix}
\]
and 
\[
\mu=
\begin{pmatrix}
1 & 0 & 0 & -1 & 0 \\
0 & 1 & 0 & -1 & 0 \\
0 & 0 & 1 & -1 & 0 \\
0 & 0 & 0 & 1 & 0 \\
-1 & -1 & -1 & 2 & 1
\end{pmatrix}.
\]
That gives us
\begin{align}
p&\equiv 
\begin{pmatrix}
p(\Lambda_{(1,2)(3)}) \\
p(\Lambda_{(1,3)(2)}) \\
p(\Lambda_{(2,3)(1)}) \\
p(\Lambda_{(1,2,3)}) \\
p(\Lambda_{(1)(2)(3)})
\end{pmatrix}\\
&=
\begin{pmatrix}
\EV{P_{1,2}}-\EV{P_{1,2,3}} \\
\EV{P_{1,3}}-\EV{P_{1,2,3}} \\
\EV{P_{2,3}}-\EV{P_{1,2,3}} \\
\EV{P_{1,2,3}} \\
1-\EV{P_{1,2}}-\EV{P_{1,3}}-\EV{P_{2,3}}+2\,\EV{P_{1,2,3}}
\end{pmatrix}.
\nonumber
\end{align}
Note that, in general, it is not guaranteed that all the probability weights $p(\Lambda)$ are positive. For example, if we use a time-delay model for the distinguishability of the photons where photon pair one-two as well as one-three are time $\tau$ apart, then photons two-three are $2\tau$ apart (see Fig.~(\ref{fig:comparisionfigure})). For this distinguishability model we get
\[
p=
\begin{pmatrix}
- e^{-\frac{3\tau^{2}}{2}} + e^{-\frac{\tau^{2}}{2}} \\
- e^{-\frac{3\tau^{2}}{2}} + e^{-\frac{\tau^{2}}{2}} \\
e^{-2\tau^{2}} - e^{-\frac{3\tau^{2}}{2}} \\
e^{-\frac{3\tau^{2}}{2}} \\
1 - e^{-2\tau^{2}} + 2 e^{-\frac{3\tau^{2}}{2}} - 2 e^{-\frac{\tau^{2}}{2}}
\end{pmatrix}
\]
where the third entry
\[
e^{-2\tau^2}-e^{-\frac{3\tau^2}{2}}
= e^{-2\tau^2}\left(1 - e^{\frac{\tau^2}{2}}\right)
\le 0
\]
is non-positive for all $\tau \neq 0$.

\section{Twirling covariance of the Fourier transform witness}\label{app:TwirlingCovariance}
In this section we show why the observable $p_{\text{f}}$ in the Fourier transform witness is unaffected by permutation twirling. 
More concretely, we want to show whether
\begin{equation}
\boxed{
    p_{\text{f}}(\rho) \stackrel{?}{=} p_{\text{f}}(\mathcal{P}(\rho)) \enspace \forall \rho.
    }
\end{equation}
Our argument has the following structure. We begin by noting that $p_{\text{f}}$ is a linear function of $\braket{J_{\sigma}}_{\rho}$. We then assert under which condition such linear function is invariant under permutation twirling $\mathcal{P}(\cdot)$ and finally we show that this condition is satisfied in case of the Fourier interferometer due to its symmetric properties.

\paragraph{General linear form of $p_{\text{f}}$.}
For any fixed interferometer $U$ and any output pattern $s$, we get the output probability
\begin{equation}
    p(s|\rho)=\frac{1}{\mu(n)\mu(s)}\sum_{\sigma \in S_n}\braket{J_{\sigma}}_{\rho} K_{\sigma}(U,s),
\end{equation}
where $K_{\sigma}=\text{perm}(M\odot M^*_{\sigma,\mathbb{I}})$ ($M$ as defined in Eq.~(\ref{M})) are the permanent weights defined in Eq.~(\ref{Shchesnovich1}) for each path that contributes to the total transition probability. Now let us consider the sum over all forbidden outcomes $\mathcal{S}_{\text{forb}}$ that would have zero probability for a perfectly indistinguishable input,
\begin{equation}
    p_{\text{f}}(\rho)=\sum_{s\in \mathcal{S}_{\text{forb}}} p(s|\rho)=\sum_{\sigma \in S_n} \braket{J_{\sigma}}_{\rho} C_{\sigma}
\end{equation}
with
\begin{equation}
    C_{\sigma} := \sum_{s\in S_{\text{forb}}} \frac{1}{\mu(n)\mu(s)}K_{\sigma}(U_F,s).
\end{equation}
We note that $p_{\text{f}}$ is a linear functional of the vector $\braket{J_{\sigma}}_{\rho}$.\\

\paragraph{What permutation twirling does to $\braket{J_{\sigma}}$:}
A twirled state has the  form
\begin{equation}
    \rho \mapsto \rho'\equiv\mathcal{P}(\rho)=\frac{1}{n!}\sum_{\tau \in S_n} J_{\tau}\rho J^{\dagger}_{\tau}.
\end{equation}
This translates to the expectation value of $J_{\sigma}$ as
\begin{align}
 \nonumber
    \braket{J_{\sigma}}_{\rho'}&= \text{tr}(\rho'J_{\sigma}) \\
    \nonumber
    &= \text{tr}\bigg(\frac{1}{n!}\sum_{\tau \in S_n}J_{\tau}\rho J_{\tau}^{\dagger}J_{\sigma}\bigg) 
    \\
     \nonumber
    &= \frac{1}{n!}\sum_{\tau \in S_n} \text{tr}\bigg( J_{\tau}\rho J_{\tau}^{\dagger}J_{\sigma}\bigg) \\
     \nonumber
    &= \frac{1}{n!}\sum_{\tau \in S_n} \text{tr}\bigg( \rho J_{\tau}^{\dagger}  J_{\sigma}J_{\tau} \bigg) 
    \\
     \nonumber
    &= \frac{1}{n!}\sum_{\tau \in S_n} \text{tr}\bigg( \rho J_{\tau^{-1}\sigma \tau} \bigg) \\
    &=\frac{1}{n!}\sum_{\tau \in S_n} \braket{J_{\tau^{-1}\sigma \tau}}_{\rho}.\label{twirledexpectvalue}
\end{align}
We see that twirling conjugates every $J_{\sigma}$ by all permutations and averages the result (the conjugation of $A$ refers to the operation $A \mapsto U^{\dagger}AU$, where $U$ is a unitary operation. Thus, conjugating $A$ means expressing it in a relabeled basis). The set of all permutations of the form $\tau^{-1}\sigma \tau$ is called the conjugacy class of $\sigma$. Since all elements in a conjugacy class are identical up to a change in basis (in our case they consist of all permutations that look the same up to a relabeling of the points) they share the same cyclic structure (integer partition). For example, in the case of three photons we have three conjugacy classes or cyclic structures,
\begin{itemize}
    \item \{(1)(2)(3)\} 
    \item \{(1,2)(3),(1,3)(2), (2,3)(1)\} - all 2-cycles
    \item \{(1,2,3)(1,3,2)\} - all 3-cycles.
\end{itemize}
Twirling (i.e., performing the group average) averages the $\braket{J_{\sigma}}$ over all permutations with the same cycle structure. It is in this sense that twirling enforces conjugacy-class invariance (cycle-type invariance). This is a coarse-graining: it identifies all permutations with the same cycle type. The orbit invariance discussed earlier and in \cite{annoni2025} is a sub-symmetry of cycle-type invariance.\\

\paragraph{When is a linear functional invariant under twirling?}
Consider any linear functional of $\braket{J_{\sigma}}$
\begin{equation}
f(\rho)=\sum_{\sigma}C_{\sigma}\braket{J_{\sigma}}_{\rho}.
\end{equation}
For the twirled state, this takes the  form
\begin{align}
f(\rho')&=\sum_{\sigma}C_{\sigma}\braket{J_{\sigma}}_{\rho'}\\
\nonumber
    &= \sum_{\sigma} C_{\sigma} \frac{1}{n!}\sum_{\tau}\braket{J_{\tau^{-1}\sigma \tau}}_{\rho}\\
    \nonumber
    &=\frac{1}{n!}\sum_{\sigma}\sum_{\tau}C_{\sigma}\braket{J_{\tau^{-1}\sigma \tau}}_{\rho},
    \nonumber
\end{align}
where we have used the twirled expectation value in Eq.~(\ref{twirledexpectvalue}). If we now relabel $\sigma'=\tau^{-1}\sigma \tau$ this becomes
\begin{equation}
    f(\rho')= \frac{1}{n!}\sum_{\sigma'}\sum_{\tau}C_{\tau\sigma'\tau^{-1}}\braket{J_{\sigma'}}_{\rho}.
\end{equation}
We can thus state that $f(\rho)$ is invariant under twirling iff the condition 
\begin{equation}
\boxed{
    \frac{1}{n!}\sum_{\tau}C_{\tau\sigma'\tau^{-1}} \stackrel{!}{=} \frac{1}{n!}\sum_{\tau}C_{\sigma'} \enspace \forall \sigma'
    }
\end{equation} 
is satisfied. To complete the proof, it only remains to show that $C_{\sigma}$ is a \textit{class function} on the permutation group $S_n$,  i.e.,  $C_{\tau\sigma\tau^{-1}}=C_\sigma$, in the case where the interferometer is a Fourier transformation.\\

\paragraph{Why is $C_{\sigma}$ a class function on $S_n$.}
Consider the discrete Fourier transform and its associated set of suppressed outcomes $\mathcal{S}_{\text{forb}}$.\\

\begin{lemma}[Fourier-twirling equivalence]
The coefficients $C_{\sigma}$ in the Shchesnovich expansion $p_{\text{f}}(\rho)=\sum_{\sigma}C_{\sigma}\braket{J_{\sigma}}_{\rho}$ satisfy \eqn{C_{\tau\sigma\tau^{-1}}=C_{\sigma} \enspace} for all $\sigma, \tau \in S_n$. Hence, $C_{\sigma}$ depends only on the cycle structure of $\sigma$, and
\begin{equation}
p_{\text{f}}(\rho)=p_{\text{f}}(\mathcal{P}(\rho))
\end{equation}
for all density operators $\rho$.
\end{lemma}

\begin{proof}
We begin by restating the definition of $C_\sigma(U_F,s)$ in terms of the permanent weight (or Fourier kernel) $K_\sigma(U_F,s)$
as
\begin{equation}
    C_{\sigma}(U_F,s) := \sum_{s\in S_{\text{forb}}} \frac{1}{\mu(n)\mu(s)}K_{\sigma}(U_F,s).
\end{equation}
Furthermore, in order to show that $C_{\tau\sigma\tau^{-1}}(U_F,s)=C_{\sigma}(U_F,s)$ we need the following two ingredients: 
\begin{enumerate}
    \item The set of suppressed outputs $\mathcal{S}_{\text{forb}}$ is closed under cyclic relabeling of the output modes $\zeta$ (see Ref.\ \cite{Tichy:2012gt}), i.e.,  \begin{equation}\label{closureSforb}
        s\in S_{\text{forb}} \Rightarrow s' \equiv \zeta \cdot s \in S_{\text{forb}} \enspace \forall \zeta \in 
        \mathbb{Z}_m.
    \end{equation}

    \item The image of conjugations $\tau\sigma\tau^{-1}$ spans all the permutations with the same cycle structure as $\sigma$ (see  Ref.\ \cite{annoni2025}). Another way of moving between elements within the same cycle structure is via cyclic shifts, $\zeta\in\mathbb{Z}_m$. It follows that for every (conjugating) permutation $\tau \in S_n$ there exists a corresponding cyclic shift $\zeta\in\mathbb{Z}_m$ of the output pattern $s \mapsto s'=\zeta \cdot s \in \mathcal{S}_{\text{forb}}$. The Fourier Kernel $K_\sigma$ thus transforms as follows under conjugations
    \begin{align}
    \nonumber
        K_{\sigma}(U_F,s) &\to
        K_{\tau^{-1}\sigma \tau}(U_F,s)\\
        &= K_{\sigma}(U_F,s'),\label{Kconjcovariance}
    \end{align}
    and is therefore \emph{invariant} under the simultaneous action of conjugation and cyclic shifts.
\end{enumerate}
Since $C_\sigma$ contains a sum over the set of forbidden patterns $S_\text{forb}$, it is invariant under cyclic shifts. This is because, as we have seen, cyclic shifts act merely as a re-parametrization of the set of forbidden outputs, and the sum over all of its elements is invariant under any re-parametrization. Given the two statements above, we can write
\begin{align}
C_{\tau\sigma\tau^{-1}} &= \sum_{s\in S_{\text{forb}}}\frac{1}{\mu(n)\mu(s)} K_{\tau^{-1}\sigma \tau}(U_F,s)\nonumber\\
    \nonumber
    &= \sum_{s'\in S_{\text{forb}}}\frac{1}{\mu(n)\mu(s')} K_{\sigma}(U_F,s')\\
    &=C_{\sigma} \label{classfunctionproof},
\end{align}
where in the second line we have used Eq.~(\ref{Kconjcovariance}), the invariance property of the Kernel, and in the third line we have employed  Eq.~(\ref{closureSforb}), the closure of the set of forbidden outputs under cyclic shifts. This completes the proof.
\end{proof}

\section{Probabilities of valid output patterns with OBB distinguishability}
\label{app:OBB}

In this appendix, we show the calculation of obtaining the probability of observing valid output patterns for a given $c_n$ in the \emph{orthogonal-bad-bit} (OBB) model of partial photon distinguishability. The purpose of this note is to quantify how distinguishability affects the probability of observing many-particle interference suppression laws in a Fourier interferometer. Throughout, we consider only distinguishability errors.

\subsection{Setup and zero-transmission law}
Assume that we input $n$ single photons into an $n$-mode Fourier interferometer, perform perfect photon-number-resolving detection on all $n$ output modes, and calculate the probability of obtaining a \emph{valid} output pattern, i.e., one satisfying the \emph{zero-transmission law} (ZTL). More precisely, a Fock measurement pattern $(s_0,\ldots,s_{n-1})$ is valid if
\begin{equation}
    \sum_{i=0}^{n-1} i s_i \equiv 0 \pmod n .
\end{equation}
Here, $s_i$ is the number of photons detected in mode $i$.
We will need the following elementary property of the ZTL.

\begin{lemma}[Existence of unique mode]
Let $(t_0,\ldots,t_{n-1})$ be a pattern of $n-1$ photons in $n$ modes. There is a unique mode in which a photon can be added to obtain a valid $n$-photon pattern.
\end{lemma}

\begin{proof}
If a photon is added in mode $j$, the validity condition requires
\begin{equation}
    j + \sum_{i=0}^{n-1} i t_i \equiv 0 \pmod n ,
\end{equation}
or equivalently
\begin{equation}
    j \equiv -\sum_{i=0}^{n-1} i t_i \pmod n .
\end{equation}
Since $j \in \{0,\ldots,n-1\}$, this uniquely specifies the value of $j$ required to obtain a valid $n$-photon pattern.
\end{proof}

\subsection{Single fully distinguishable photon}

We introduce notation for internal degrees of freedom. Let $\{\xi_j\}$ be an orthonormal basis for the space of internal states, and let $a_i^\dagger[\xi_j]$ denote the creation operator for a photon in spatial mode $i$ with internal state $\xi_j$.
We will consider state vectors $\ket{\psi}$ involving both internal and external degrees of freedom. By the \emph{external state} corresponding to $\ket{\psi}$ we mean the state obtained by tracing out the internal modes.

\begin{theorem}[Probability of obtaining a valid pattern from Fourier interferometry]
Let $\ket{\psi}$ be the  input state vector of $n$ photons in $n$ modes
with the photon in mode $r$ orthogonal to all the others,
\begin{equation}
    \ket{\psi}=a_0^\dagger[\xi_{i_0}]\cdots a_{n-1}^\dagger[\xi_{i_{n-1}}]\ket{\vec{0}},
\end{equation}
where for some $r$ we have $\langle \xi_{i_r} | \xi_{i_j} \rangle = 0$ for $j \neq r$. Then the probability of obtaining a valid pattern from the Fourier interferometer is exactly $1/n$.
\end{theorem}

\begin{proof}
First, since creation operators commute, we may express
\begin{equation}
    \ket{\psi}=a_r^\dagger[\xi_{i_r}]a_0^\dagger[\xi_{i_0}]\cdots\widehat{a_r^\dagger[\xi_{i_r}]}\cdots a_{n-1}^\dagger[\xi_{i_{n-1}}]\ket{\vec{0}},
\end{equation}
where the hat indicates that the $r$th factor is omitted.
Applying the Fourier transform yields
\begin{eqnarray}
    F_n \ket{\psi}&=&
    \left(\frac{1}{\sqrt{n}}\sum_j \omega^{-rj} a_j^\dagger[\xi_{i_r}]\right)\\
    &\times &
    \left(F_n a_0^\dagger[\xi_{i_0}]\cdots\widehat{a_r^\dagger[\xi_{i_r}]}\cdots a_{n-1}^\dagger[\xi_{i_{n-1}}]\ket{\vec{0}}\right),
    \nonumber
\end{eqnarray}
where $\omega := e^{2\pi i/n}$.
Expanding the second factor gives
\begin{equation}
    F_n a_0^\dagger[\xi_{i_0}]\cdots\widehat{a_r^\dagger[\xi_{i_r}]}\cdots a_{n-1}^\dagger[\xi_{i_{n-1}}]\ket{\vec{0}}=\sum_s c_s \ket{\psi_s},
\end{equation}
where $s$ ranges over Fock patterns of $n-1$ photons in $n$ modes, $\ket{\psi_s}$ has external state vector $\ket{s}\!\bra{s}$, and $\sum_s |c_s|^2 = 1$.
The $n$-photon state after the Fourier transform is
\begin{equation}
    F_n \ket{\psi}=\frac{1}{\sqrt{n}}\sum_s c_s\sum_j \omega^{-rj}a_j^\dagger[\xi_{i_r}]\ket{\psi_s}.
\end{equation}
For a fixed $(n-1)$-photon pattern $s$, exactly one value of $j$ yields a valid $n$-photon pattern; call it $j_s$. The contributing terms are
\begin{equation}
    \frac{1}{\sqrt{n}}\sum_s c_s \omega^{-r j_s}a_{j_s}^\dagger[\xi_{i_r}]\ket{\psi_s}.
\end{equation}
Since the added photon is distinguishable, these terms are mutually orthogonal, so that
\begin{equation}
    \langle \psi_t |a_{j_t}[\xi_{i_r}]a_{j_s}^\dagger[\xi_{i_r}]| \psi_s \rangle= \delta_{s,t}.
\end{equation}
Their total weight is, therefore,
\begin{equation}
    \frac{1}{n}\sum_s |c_s|^2=\frac{1}{n}.
\end{equation}
\end{proof}

\subsection{Probabilities of valid patterns in the OBB model}

This result applies directly to the OBB model. It also applies more generally to any unitary whose matrix elements have equal modulus $1/\sqrt{n}$, such as Hadamard interferometers or tensor-product Fourier transforms.

\begin{corollary}[Application of the theorem]
Consider a state of $n$ single photons in different modes, each of the form
\(
\rho = (1-\epsilon)\,\rho_{\mathrm{id}} + \epsilon\,\rho_{\mathrm{dist}}
\)
in the OBB model. Then the probability of obtaining a valid pattern is
\begin{equation}
    \left(1-\frac{1}{n}\right)(1-\epsilon)^{n} + \frac{1}{n}.
\end{equation}
\end{corollary}

\begin{proof}
As usual, we expand
\begin{equation}
    \rho^{\otimes n}
    =
    \sum_{k=0}^n
    \binom{n}{k}
    \epsilon^k (1-\epsilon)^{n-k}\Phi_k,
\end{equation}
where $\Phi_k$ is the mixed state of all terms with $k$ distinguishable photons.
Let $h(\eta)$ be the probability of obtaining a valid pattern with starting state $\eta$. For $k\ge1$, $h(\Phi_k)=1/n$, while
$h(\Phi_0)=1$. Then
\begin{align}
    h\!\left(\rho^{\otimes n}\right)
    &= \sum_{k=0}^n
    \binom{n}{k}
    \epsilon^k (1-\epsilon)^{n-k} h(\Phi_k)\\
    \nonumber
    &= (1-\epsilon)^{n}
    + \frac{1}{n}\sum_{k=1}^n
    \binom{n}{k}
    \epsilon^k (1-\epsilon)^{n-k} \\
     \nonumber
    &= \left(1-\frac{1}{n}\right)(1-\epsilon)^{n}
    + \frac{1}{n}\sum_{k=0}^n
    \binom{n}{k}
    \epsilon^k (1-\epsilon)^{n-k}\\
     \nonumber
    &= \left(1-\frac{1}{n}\right)(1-\epsilon)^{n} + \frac{1}{n}.
\end{align}
\end{proof}

\paragraph*{Sampling complexity for the Fourier witness.} $h\!\left(\rho^{\otimes n}\right)\equiv p_{\text{valid}}$ is the probability of obtaining a valid pattern with an $n$ photon state $\rho^{\otimes n}$. Therefore, we can write the probability for forbidden patterns as $p_{\text{f}}=1-p_{\text{valid}}$. Furthermore, what we called $(1-\epsilon)^n$ in the previous Corollary is precisely the indistinguishable weight $c_n$. Thus, we can write
\begin{align}
    p_{\text{f}} &=1-p_{\text{valid}}\\
    &= 1- \bigg[\left(1-\frac{1}{n}\right)c_n + \frac{1}{n}\bigg].
    \nonumber
\end{align}
Solving for $c_n$ in terms of $p_{\text{f}}$ yields
\begin{equation}
    c_n = 1-\frac{n}{n-1}p_{\text{f}}
\end{equation}
which has the form of Eq.~(\ref{eq:lambda}), with $\lambda=(n-1)/{n}$.

\section{Sample complexity}\label{app:sample_complexity}
In Section~\ref{comparison} the sample complexity of the witnesses was compared, in this appendix the calculations to compare the sample complexity is elaborated on.
We start with Hoeffding's inequality \cite{Hoeffding1963}. Hoeffding's inequality expresses the largest deviation ($t$) in measured quantities ($X_n$) as a function of the error probability ($\epsilon$) and the number of samples ($N$). Therefore, Hoeffding's bounds are used to bound the fidelity. 

\begin{equation}
\mathbb{P}\left( S_N -  E\left[S_N\right]  \ge t \right) \le \exp \left( \frac{-2 t^2}{\sum_{i=1}^N (b_i - a_i)^2} \right) .
\end{equation}
 
To find the amount of deviation that could be expected for a confidence, an error probability is defined as $\epsilon$. The chance of measuring more than the maximal difference ($\delta$) between the mean probability ($p_x$) from the probability density function that was measured with $N$ samples and the expected value of the mean probability $(E[p_x]$) is lower than this error probability. This is expressed as
\begin{equation}
    P\left(p_x - \mathbb{E}[p_x] > \delta \right) \leq \epsilon.
\end{equation}
The probability is obtained from samples $S_N$ by calculating the probabilty $p_x = E\left[S_N\right]/N$ such that
\begin{equation}
    P\left( \frac{S_N}{N} -  \frac{E\left[S_N\right]}{N} > \delta \right) \leq \epsilon.
\end{equation}
Both sides of the inequality are multiplied by the amount of samples ($N$),
\begin{equation}
    P\left( S_N -  E\left[S_N\right] > \delta N \right) \leq \epsilon.
\end{equation}
From Hoeffding's inequality, we can substitute with $t=N\delta$, to get
\begin{equation}
    \epsilon = \exp \left( \frac{-2 (N\delta)^2}{\sum_{i=1}^N (b_i - a_i)^2} \right) .
\end{equation}
Since the variable we are measuring is probability, we know that $0 \le p_x \le 1$. So,  $\sum_{i=1}^N (b_i - a_i)^2 = N$. This leads to the expression
\begin{equation}\label{eq: epsi is}
    \epsilon = \exp \left( -2 N\delta^2 \right).
\end{equation}
From Eq.~(\ref{eq: epsi is}), we obtain the maximum error 
\begin{equation} \label{eq: delta hoef}
    \delta(N, \epsilon) = \sqrt{\frac{\ln(\frac{1}{\epsilon})}{2N}}
\end{equation}
 in the probabilities that can be added or subtracted from the probabilities that are used for the bounds on the indistinguishability and the fidelity. Now that this error in measured probabilities is defined, it must be used to calculate correct bounds on the fidelity and then also on the indistinguishability. Theoretically, the fidelity can be calculated as 
\begin{equation}
    F \leq p_{U^\dagger U} - (1 - c_n),
\end{equation}
and the fidelity becomes
\begin{equation}
    F \leq (p_{U^\dagger U} - \delta_{p_{U^\dagger U}}(N,\epsilon)) - c_n.
\end{equation}
The indistinguishability must also be corrected. This correction is dependent on the method used to determine $c_n$. The actual lower bound $c_n(\epsilon)$ is a function of the probability for the bound to be valid. An error parameter $\xi$ is subtracted from the measured indistinguishability parameter, which will be denoted by $c_n'$. This error parameter is factorized into the method-dependent factor ($\kappa_n$) and the factor obtained from Hoeffding's bounds ($\delta(N,\epsilon)$)
is
\begin{equation}
    c_n(\epsilon) \geq c_n' - \xi(N, \epsilon) = c_n' - \kappa_n \delta(N, \epsilon).
\end{equation}
$\kappa_n$ is dependent on the method  and the number of photons $n$ that are used, and $\delta(N, \epsilon)$ is given in Eq.~(\ref{eq: delta hoef}). 

The next step is to calculate the method-dependent factors for the different methods of witnessing the indistinguishability. This is done by adding or subtracting the error from the Hoeffding's bounds directly to the measured probabilities. This is systematically done for all four witnesses.

For the Fourier method, the correction is added to the probability of measuring a forbidden mode
\begin{equation}
    c_n > 1 - \frac{n}{n-1}(p_{\text{f}} + \delta) = c_n' - \left(1 + \frac{1}{n-1}\right)\delta,
\end{equation}
with $n$ as the number of photons. The correction factor is found to be
\begin{equation}\kappa_{n,\text{Fourier}}=\frac{n}{n-1}. 
\end{equation}
This means that the scaling is asymptotically approaching 1 for big $n$. For the $n=3$ case $\kappa_{3,\text{Fourier}}=\frac{3}{2}$.



For the cyclic interferometer method, the correction is added to the probability of the positive fringe and deducted from the negative such that there is a total penalty of $2\delta$ given by
\begin{equation}
    c_n > 2^{n-1}(|p_+ - p_-| - 2\delta) = c_n'- 2^n\delta.
\end{equation}
It can be seen that $\kappa_{n,\text{Pont}}=2^n$ which is exponentially growing with the number of photons. For the 
$3$-photon case, the 
pre-factor is $8$.

For the superposition HOM method, we subtract the error from the bunching probability
\begin{equation}
    c_n > (2n-2)(p_{\text{b}}-\delta) - 2n +3= c_n' - (2n-2)\delta.
\end{equation}
So, $\kappa_{n,\text{Brod}}=2n-2$ and for the case that will be measured  $\kappa_{3,\text{Brod}}=4$. 

For the two-mode correlator method, we use the correlator instead of the probability. We have to slightly change the definition of the error. First, we define the correlator as 
\begin{equation}
    C_{1,2} = \left< n_1 n_2 \right> - \left< n_1 \right>\left< n_2 \right>.
\end{equation}
 Then we determine the error for the number operator. The difference is that the number operators can be measured to be more than 1. So, $b_i$ equals the number of photons in the system. However, for $\left< n_1 n_2\right>$ the max is $(\frac{n}{2})^2$ for even and $\frac{n+1}{2}\frac{n-1}{2}$ 
 for odd $n$. The bound for the 
 correlator would be
\begin{equation}
    C_{1,2} > \left< n_1 n_2 \right> - \delta_{\left< n_1 n_2 \right>} - (\left< n_1 \right>-\delta_{\left< n_1 \right>})(\left< n_2 \right> -\delta_{\left< n_2 \right>}),
\end{equation}
with delta as the error defined with $b_i=1$. Then, the indistinguishability witness becomes
\begin{widetext}

\begin{equation}
 c_n  >
    \begin{cases}
        \frac{4n}{n-2}(\left< n_1 n_2 \right> - (\frac{n}{2})^2 \delta - (\left< n_1 \right> + n\delta)(\left< n_2 \right> + n\delta)),    &  \text{for } n=\text{even}, \\
        \frac{4n}{n-2}(\left< n_1 n_2 \right> - \frac{n-1}{2}\frac{n+1}{2} \delta - (\left< n_1 \right> + n\delta)(\left< n_2 \right> + n\delta)),    & \text{for } n=\text{odd}.
    \end{cases}
\end{equation}
\end{widetext}
This equation does 
not lend itself to the pre-factor notation. However, it is already clear that the bound grows polynomially (with $O(n^2)$), but also that the value is much larger than the others for the three-photon case. To be exact $n=3$ is given by
\begin{equation}
    c_n > 12(\left< n_1 n_2 \right> - 2 \delta - (\left< n_1 \right> + 3\delta)(\left< n_2 \right> + 3\delta)).
\end{equation}
To summarize the information, the pre-factors are collected in Table \ref{tab: pre-factors}. These pre-factors are directly related to the sampling complexity. If the pre-factor is larger, it will take more samples to get the desired bound tightness and certainty. This table is also presented in the 
main text.

\begin{table}[H] 
    \centering
    \begin{tabular}{@{}lcc@{}}
    \hline \hline
    Method  & Generalized pre-factor ($\kappa_n$)  & pre-factor for $n=3$ ($\kappa_3$)     \\ \hline
    Fourier  transform            & $\frac{n}{n-1}$  & $\frac{3}{2}$\\
    Cyclic interferometer                     & $2^n$           & 8       \\
    Superposition HOM  dip                     & $2n-2$         & 4         \\
    Two-mode correlator                   & $O(n^2)$        &  $ > 24$    \\
    \hline \hline
    \end{tabular}
    \caption{The pre-factors $\kappa_n$ for the calculation of the Hoeffding's bounds with $c_n(\epsilon) > c_n' - \kappa_n\delta(N, \epsilon)$ for the different witnesses. 
    The derivations can be found in Appendix \ref{app:sample_complexity}. For the two-mode correlator method, this calculation is not possible, but the scaling is $O(n^2)$.} \label{tab: pre-factors}
\end{table}


\begin{thebibliography}{106}
\bibitem [{par()}]{part_note}%
  \BibitemOpen
  \href@noop {} {}\bibinfo {note} {E.g. the anticlockwise cyclic permutation is
  written (1,3,2) in cycle notation while the clockwise cycle is (1,2,3) but
  both permutations have the same partition $\underline{\sigma}_{(1,2,3)} =
  \underline{\sigma}_{(1,3,2)}$ which could equally well be written as
  partition as (1,2,3), (1,3,2), (3,2,1) etc.}\BibitemShut {Stop}%
\end{thebibliography}

%

\end{document}